\documentclass{amsart}

\usepackage{qzart}
\usepackage{amsaddr}
\usepackage{multirow}

\usepackage[dvipsnames]{xcolor}
\newcommand{\rev}[1]{{\color{Black} #1}}
\newcommand{\revv}[1]{{\color{Black} #1}}

\begin{document}
\title{Selective Inference for Effect Modification via the Lasso}
\author{Qingyuan Zhao, Dylan S. Small, Ashkan Ertefaie}
\address{University of Cambridge, University of
  Pennsylvania and University of Rochester}
\email{qyzhao@statslab.cam.ac.uk}
\date{\today}
\thanks{We than Jonathan Taylor and
  Xiaoying T.\ Harris for helpful discussion.}
\maketitle

\begin{abstract}
  Effect modification occurs when the effect of the treatment
on an outcome varies according to the level of other covariates and often has
important implications in decision making. When there
are tens or hundreds of covariates, it becomes necessary to use the
observed data to select a simpler model for effect modification and
then make valid statistical inference. We propose a two stage
procedure to solve this problem. First, we use Robinson's
transformation to decouple the nuisance parameters from the treatment
effect of interest and use machine learning algorithms to estimate the nuisance parameters. Next, after plugging in the estimates of the
nuisance parameters, we use the Lasso to choose a low-complexity model for
effect modification. Compared to a full model consisting of all the
covariates, the selected model is much more interpretable. Compared
to the univariate subgroup analyses, the selected model greatly reduces
the number of false discoveries. We show that the conditional
selective inference for the selected model is asymptotically valid given
the rate assumptions in classical semiparametric regression. Extensive
simulation studies are conducted to verify the asymptotic results and
an epidemiological application is used to demonstrate the method.


\end{abstract}

\section{Introduction}
\label{sec:introduction}

When analyzing the causal effect of an intervention, effect
modification (or treatment effect heterogeneity) occurs when the magnitude of the causal effect varies as
a function of other observed covariates. Much of the causal inference
literature focuses on statistical inference about the average causal
effect in a population of interest. However, in many applications, it
is also important to study how the causal effect varies in different
subpopulations, i.e.\ effect modification, for reasons
including personalizing treatments in precision medicine
\citep{murphy2003optimal,ashley2015precision},
generalizing the causal
finding to different populations \citep{stuart2011use}, better
understanding of the causal mechanism \citep[page
96]{grobbee2009clinical}, and making inference less sensitive to
unmeasured confounding bias
\citep{hsu2013effect,ertefaie2017quantitative}.

A natural way of identifying effect modification is \emph{subgroup analysis},
in which observations are stratified based on the
covariates. More generally, one can study effect modification by
including interaction terms between the treatment and the covariates
in an outcome regression. Depending on whether the subgroups or
interactions are chosen before or after any examination of the data,
the statistical analysis can be \emph{a priori} or \emph{post hoc}
\citep{wang2007statistics}.
A priori subgroup analyses are free of selection bias and are
frequently used in clinical trials and other observational
studies. They do discover some effect modification, often convincingly, from
the data, but since the potential effect modifiers are determined a
priori rather than using the data, many real effect modifiers may
remain undetected.

With the enormous amount of data and covariates
being collected nowadays, discovering effect modification by post hoc
analyses has become a common interest in several applied fields,
including precision medicine \citep{lee2016discovering,pickkers2017towards,jiang2018genome},
education \citep{schochet2014understanding}, political science
\citep{imai2013estimating,grimmer2017estimating}, economics
\citep{angrist2004treatment}, and online digital
experiments \citep{taddy2016nonparametric}. Post hoc analysis was
originally treated as a multiple comparisons problem in the works of
\citet{tukey1949comparing} and \citet{scheffe1953method}, where a
single categorical effect modifier is considered. However, modern
applications can
easily have tens or hundreds of potential effect modifiers. In this
case, it is impractical to
consider the subgroups exhaustively (for example, there are $2^{30} >
10^{9}$ distinct subgroups with just $30$ binary
covariates). Therefore, there has been an increasing demand for more
sophisticated post hoc methods that can discover effect modification from
the data.

\citet{tukey1991philosophy} discussed thoroughly about the difference
between a priori and post hoc analyses and why they are both
important. In page 103, he wrote:
\begin{quote}
Exploration---which some might term ``data dredging''---is quite
different from ``exogenous selection of a few comparisons''. Both have
their place. We need to be prepared to deal with either.
\end{quote}
For Tukey, data dredging (post hoc subgroup analyses in our
context) seems to mean a \emph{full} exploration of the data, so it
can only give interesting hints rather than confident conclusions. In
this paper we will present an effort based on the recently developed
selective inference framework \citep{taylor2015statistical} to
\emph{partially} explore the data and
then make confident conclusions about the interesting hints. Such an
approach was perhaps not
envisioned by Tukey, but he did point out that the degree of
exploration can impact the inference, which is a fundamental idea for
selective inference. In fact, confidence intervals using
his studentized range \citep{tukey1949comparing} are tighter than
Scheff\'{e}'s method because fewer (but more important, as he argued)
comparisons are made.


\subsection{Defining the effect modification problem}
\label{sec:prev-appr-effect}


Depending on the goal of the particular application, there are several ways to
set up the inferential problem of effect modification. One way to
phrase this problem is that we would like an optimal or near-optimal
treatment assignment rule for future experimental objects (e.g.\
patients). Thus the statistical problem is to estimate the probability
that a certain treatment outperforms other treatment(s) given
characteristics of the objects and quantify its uncertainty. This
\emph{optimal treatment regimes} problem has been studied extensively by
\citet{murphy2003optimal,robins2004optimal,manski2004statistical,hirano2009asymptotics,zhao2012estimating,zhang2012robust,ertefaie2014constructing,luedtke2016super,athey2017efficient}
among many others. A closely related but less direct formulation is to
estimate the \emph{conditional average treatment effect} and provide
confidence/credible intervals, see
\citet{robins2008higher,hill2011bayesian,balzer2016targeted,athey2016generalize}
for some recent references.

A third way to define the effect modification problem, which is the
formulation we will take in this paper, is to post questions like: Is
covariate A an effect
modifier? What are the other potential effect modifiers and how
certain are we about them? These questions are frequently raised
in practice to test existing scientific hypotheses, generate new
hypotheses, and gather information for intelligent decision
making. Unfortunately, the previous two approaches are not designed to
answer these questions directly. Although there exist several
proposals of discovering (in Tukey's language, interesting hints about) effect modifying covariates or subgroups
\citep{imai2013estimating,tian2014simple,athey2016recursive}, little
attention has been given to providing inference such as confidence interval
or statistical significance for the selected covariates or subgroups
(in Tukey's language, making confident conclusions).

When giving inference after model selection, the
naive inference ignoring the data dredging is generally
over-confident, but not all applied researchers
are mindful to this danger. For example, in a book on testing and
interpreting interactions in social sciences, \citet[page
105]{aiken1991multiple} recommended to drop
insignificant interactions from the regression model especially if
they are not expected by the investigator, but the authors did not
mention that the subsequent statistical inference would be
corrupted. Such a suggestion can also be found in another highly cited
book by \citet[page 361]{cohen2003applied}:
\begin{quote}
We then use the products of these main
effects collectively to test the significance of the interaction,
\ldots. If the interaction term turns out to be
significant, then the regression coefficients from the full model
including the interaction should be reported.
\end{quote}
Empirical studies that ignore the danger of cherry-picking the
interaction model can be found even in top medical journals
\citep[e.g.][]{sumithran2011long,zatzick2013randomized}. Other books
such as \citet[Section 10.3.1]{weisberg2005applied} and \citet[Section
5.3.3]{vittinghoff2011regression} warned that ``significance (in the
selected model by the naive inference) is overstated'' and ``exploratory
analyses are susceptible to false-positive findings'', but no
practical solution was given. A temporary fix used in practice is to
test each potential
interaction separately with all the main effects
\citep[e.g.][]{ohman2017clinically}, but this method may find too many
inconsequential covariates which are merely correlated with the actual
effect modifiers.

To fill this gap, we propose
a method that combines classical semiparametric regression
\citep[c.f.][]{robinson1988root,van2000asymptotic,li2007nonparametric}
with recent advances in post-selection inference for high-dimensional regression
\citep[c.f.][]{berk2013valid,lee2013exact,fithian2014optimal,tian2015selective,rinaldo2016bootstrapping}.
The basic idea is that the semiparametric regression is used to remove
confounding and increase the power of discovering treatment effect variation,
and the post-selection inferential tools are then used to ascertain
the discovered effect modifiers. Although this idea is generally
applicable to other selective inference methods, in this paper we will
focus on the selective inference framework developed in \citet{lee2013exact}
where the effect modification model is selected by the lasso
\citep{tibshirani1996regression}.

Our method can also be used in non-causal regression problems when the
interest is on the interaction terms between some low-dimensional
primary variables (not necessarily a treatment) and some high-dimensional
covariates. Examples include gene-environment interaction
\citep{crossa2010prediction,biswas2014detecting}, interaction between species
traits and environment in ecology \citep{brown2014fourth},
and subgroup-specific reactions to political issues when modeling electoral
choices \citep{mauerer2015modeling}. In these examples, the applied
researchers often want to make statistical inference for a selected
interaction model.

In the rest of the Introduction, we shall give a brief overview of
our proposal with minimal mathematical details. Since our
proposal of combining Robinson's transformation and selective
inference is very straightforward, a practitioner should be able to
implement this procedure after reading the Introduction and may
skip the theoretical details in the next few Sections until the
applied example in \Cref{sec:appl-obes-elev}.

\subsection{Causal model and the tradeoff between accuracy and interpretability}
\label{sec:causal-model-effect}

We first describe the causal model used in this paper. Suppose we observe
i.i.d.\ variables $\{\mathbf{X}_i,T_i,Y_i\}_{i=1}^n$ where the vector
$\mathbf{X}_i \in \mathcal{X} \subset \mathbb{R}^p$ contains
covariates measured before the treatment, $T_i \in \mathcal{T}$ is the
treatment assignment, and $Y_i \in \mathbb{R}$
is the observed continuous response. Let $Y_i(t)$ be the potential
outcome (or counterfactual) of $Y_i$ if
the treatment is set to $t \in \mathcal{T}$. Throughout this paper we
shall assume $Y_i = Y_i(T_i)$ (consistency of the observed outcome) and the usual unconfoundedness and positivity assumptions in
causal inference; see \Cref{sec:infer-fixed-model} for more
detail. We allow our dataset to come from a randomized
experiment (where the distribution $T_i|\mathbf{X}_i = \mathbf{x}$ is known)
or an observational study (where the distribution $T_i|\mathbf{X}_i =
\mathbf{x}$ must be estimated from the data).

We assume a nonparametric model for the potential outcomes,
\begin{equation}
  \label{eq:causal-model}
  Y_i(t) = \eta(\bm{X}_i) + t \cdot \Delta(\bm{X}_i) + \epsilon_i(t),~i=1,\dotsc,n.
\end{equation}
Here $\eta$ and $\Delta$ are functions defined on $\mathcal{X}$ and
$\mathrm{E}[\epsilon_i(t)|\mathbf{X}_i] = 0$. Our model
\eqref{eq:causal-model} is very general. It is in fact saturated if
the treatment $T_i$ is binary, i.e.\ $\mathcal{T} = \{0, 1\}$. In this
case, $\Delta(\mathbf{x}) = \mathrm{E}[Y_i(1) - Y_i(0) |
\mathbf{X}_i = \mathbf{x}]$ is commonly referred to as the conditional average
treatment effect (CATE). When the treatment is continuous, i.e.\ $\mathcal{T}
\subseteq \mathbb{R}$ (for example
dosage), model \eqref{eq:causal-model} assumes the interactions between the
treatment and the covariates are
linear in the treatment but possibly nonlinear in the covariates.
In causal inference, $\Delta(\mathbf{x})$ is the parameter of interest
(for example, $\mathrm{E}[\Delta(\mathbf{X})]$ is the average treatment
effect), whereas $\eta(\mathbf{x})$ is regarded as an
infinite-dimensional nuisance parameter. We say there is \emph{effect
  modification} if the
function $\Delta(\mathbf{x})$ is not a constant.

When the dimension of $\mathbf{X}$ is high, there is a fundamental tradeoff
between estimation accuracy and interpretability of
$\Delta(\bm{x})$. On one extreme (the rightmost column in
\Cref{tab:tradeoff}), one could use flexible
machine learning methods to estimate $\Delta(\bm{x})$, which is
important when the goal is accurate prediction (e.g.\ for the purpose of
personalizing the treatment). However, such models are generally very difficult to
interpret \citep[see e.g.][]{zhao2017causal}. For example, in random
forests it is challenging to even define a good notion of
variable importance \citep{strobl2007bias}. The full linear
model $\Delta(\bm{x}) \approx \alpha+\bm{x}^T\bm{\beta}$ suffers
from the same interpretability problem when the dimension of $\bm{x}$
is more than just a few. 
Moreover, important effect modifiers may be
masked by noise covariates.

On the other extreme (the leftmost column in \Cref{tab:tradeoff}), one
could run marginal regressions to test if (estimated) $\Delta(\bm{x})$ is
correlated with each covariate. However, this method usually discovers too
many false positive covariates in the sense that they are no longer
associated with $\Delta(\bm{x})$ after conditioning on the actual
effect modifiers. For example, in our example in
\Cref{sec:appl-obes-elev}, the most probable effect
modifiers are gender and age. However, univariate regressions also
find several other covariates significant, including marital
status and whether the person had arthritis, heart attack, stroke, and
gout. This is most likely due to their strong correlation with age.

Following the ideas in \citet{berk2013valid} and \citet{lee2013exact}, in this
paper we propose to use a linear submodel to
approximate the treatment effect, $\Delta(\bm{x}) \approx
\alpha_{\hat{\mathcal{M}}}+\bm{x}_{\hat{\mathcal{M}}}^T\bm{\beta}_{\hat{\mathcal{M}}}$,
where the subset $\hat{\mathcal{M}} \subseteq \{1,\dotsc,p\}$ is
selected using the data. We argue that a low-dimensional linear model
often yields the desirable tradeoff between accuracy and
interpretability by selecting important effect modifiers,
especially if the goal is to discover a few effect modifiers for
further investigation or to select covariates to personalize
the treatment. In our applied example \Cref{sec:appl-obes-elev} that
considers the effect of being overweight on systematic inflammation, our method
selected the linear submodel with gender, age, stroke, and gout to
approximate $\Delta(\bm{x})$. After adjusting for model selection,
stroke and gout become non-significant in the submodel. \Cref{tab:tradeoff}
gives a comparison of the
strengths and weaknesses of using different statistical models to approximate
$\Delta(\bm x)$.

\begin{table}[t]
  \centering
  \renewcommand{\arraystretch}{1.2}
  \begin{tabular}{|c|c|c|c|c|}
    \hline
    & Univariate model & {\bf Selected submodel} & Full linear model &
    Machine learning \\
    \hline
    Model of $\Delta(\bm{X}_i)$ & $\alpha_j + X_{ij}^T \beta_j$ &
    $\alpha_{\hat{\mathcal{M}}}+\bm{X}_{i,\hat{\mathcal{M}}}^T\bm{\beta}_{\hat{\mathcal{M}}}$
    & $\alpha+\bm{X}_i^T\bm{\beta}$ & e.g.\ additive trees \\
    \hline
    Accuracy & Poor & {\bf Good} & Good & Very good \\
    \hline
    Interpretability & Very good & {\bf Good} & Poor & Very poor \\
    \hline
    \multirow{2}{*}{Inference} & Easy, but many & {\bf Need to consider}
    & Semiparametric & \multirow{2}{*}{No clear objective} \\
    &  false positives & {\bf model selection} & theory & \\
    \hline
  \end{tabular}
  \caption{Tradeoff of accuracy and interpretability of different
    models of effect modification. In the case of high-dimensional
    covariates, machine learning and full linear model approximate
    $\Delta(\mathbf{x})$ more accurately but are difficult to
    interpret. Univariate regressions find the covariates correlated
    with $\Delta(\mathbf{x})$ but may end up with false positives that
  are no longer correlated with $\Delta(\mathbf{x})$ after
  conditioning on other covariates. The selected submodel approach
  proposed in this paper is an attractive trade-off between accuracy
  and interpretability. See \Cref{sec:when-should-select} for more discussion.}
  \label{tab:tradeoff}
\end{table}

\subsection{Our proposal}
\label{sec:our-proposal}

In this paper we will select the linear submodel of effect modification
using the lasso \citep{tibshirani1996regression}, which has been
shown to be very effective at selecting relevant variables in
high-dimensional regression
\citep{zhao2006model,hastie2009elements,buhlmann2011statistics}. To
illustrate the proposal, for a moment let's
assume $\eta(\bm{x}) = 0$ for all $\bm x \in \mathcal{X}$ and $T_i
\equiv 1$, so \eqref{eq:causal-model} becomes a conditional
mean model:
\begin{equation} \label{eq:illustration-model}
Y_i = \Delta(\bm{X}_i) + \epsilon_i,~i=1,\dotsc,n.
\end{equation}
We can select a small linear submodel by running the following lasso
regression with a prespecified regularization parameter $\lambda$,
\begin{equation} \label{eq:illustration-lasso}
\underset{\alpha,\bm{\beta}}{\mathrm{minimize}}~ \sum_{i=1}^n (Y_i -
\alpha - \mathbf{X}_i^T \bm{\beta})^2 + \lambda \|\bm{\beta}\|_1.
\end{equation}
Let the selected model $\hat{\mathcal{M}} \subset
\{1,\dotsc,p\}$ be the positions of non-zero entries in the solution
to the above problem.
\citet{lee2013exact} derived an exact inference of the regression parameter
$\bm{\beta}^{*}_{\hat{\mathcal{M}}}$ where
$\alpha^{*}_{\hat{\mathcal{M}}}+\bm{X}_i^T\bm{\beta}^{*}_{\hat{\mathcal{M}}}$
  is the ``best submodel approximation'' of $\Delta(\bm{X}_i)$ in
  Euclidean distance\rev{, assuming the design matrix is fixed}. Thus \rev{
  the parameter $\bm{\beta}^{*}_{\hat{\mathcal{M}}}$ here depends on
  the randomness in $\bm X$ in a random design setting. Furthermore,}
$\bm{\beta}^{*}_{\hat{\mathcal{M}}}$ also depends on the selected
model $\hat{\mathcal{M}}$ \rev{and thus inherits the randomness in the
  model selection procedure}\citep{berk2013valid}.

  Based on a pivotal statistic
  obtained by \citet{lee2013exact} and by conditioning on the
  selection event, it is possible to form valid confidence intervals
  for the entries of $\bm{\beta}^{*}_{\hat{\mathcal{M}}}$ adjusting
  for the fact that the submodel
$\hat{\mathcal{M}}$ is selected using the data. For example,
we can find a selective confidence interval $[D_j^{-},D_j^{+}]$ for
the $j$-th element of $\bm{\beta}^{*}_{\hat{\mathcal{M}}}$ that satisfies
\begin{equation} \label{eq:stie}
  \mathrm{P}\Big(\big(\bm{\beta}^{*}_{\mathcal{M}}\big)_j \in
  [D_j^{-},D_j^{+}]\,\Big|\,\hat{\mathcal{M}} = \mathcal{M}\Big) = 1 - q.
\end{equation}
An important consequence is that this guarantees the
control of false coverage rate (FCR), that is
\begin{equation} \label{eq:fcr}
\mathrm{E}\bigg[\frac{\#\big\{1 \le j \le |\hat{\mathcal{M}}|:
    \big(\bm{\beta}^{*}_{\hat{\mathcal{M}}}\big)_j \not\in
    [D^{-}_j,D^{+}_j]\big\}}{\max(|\hat{\mathcal{M}}|,1)} \bigg] \le q.
\end{equation}
FCR is the average proportion of non-covering
confidence intervals and extends the concept of false discovery rate
to estimation \citep{benjamini2005false}. Notice that \eqref{eq:stie}
directly implies \eqref{eq:fcr} given
$\hat{\mathcal{M}} = \mathcal{M}$, and \eqref{eq:fcr} can then be proven by
marginalizing over $\hat{\mathcal{M}}$. See \citet[Lemma
2.1]{lee2013exact} and \citet{fithian2014optimal}.

The main challenge to directly applying selective inference
to the effect modification problem is the
nuisance parameter $\eta(\bm{x})$. In this paper we propose to use the
technique in \citet{robinson1988root} to eliminate this nuisance
parameter. Our proposal is a two-stage procedure. In the first stage,
we introduce two nuisance parameters:
$\mu_y(\mathbf{x}) = \mathrm{E}[Y_i|\mathbf{X}_i = \mathbf{x}]$ and
$\mu_t(\mathbf{x}) = \mathrm{E}[T_i|\mathbf{X}_i =
\mathbf{x}]$, so $\mu_y(\mathbf{x}) = \eta(\mathbf{x}) +
\mu_t(\mathbf{x}) \Delta(\mathbf{x})$ by
\eqref{eq:causal-model} and unconfoundedness. The nonparametric model
\eqref{eq:causal-model} can be rewritten as
\begin{equation}
  \label{eq:transformed-model}
  Y_i - \mu_y(\mathbf{X}_i) = \big[T_i - \mu_t(\mathbf{X}_i)\big] \cdot
  \Delta(\mathbf{X}_i) + \epsilon_i,~i=1,\dotsc,n.
\end{equation}
We have eliminated $\eta(\bm{x})$ from the model but introduced two
more nuisance parameters, $\mu_t(\bm{x})$ and
$\mu_y(\bm{x})$. Fortunately, these two nuisance functions can be directly
estimated by regression using the pooled data, preferably using some
machine learning methods with good prediction performance as advocated
by \citet{vanderlaan2011} and \citet{chernozhukov2016double} (see
\Cref{sec:assumptions-paper} for more discussion). In the numerical
examples, we will use the random forests \citep{breiman2001random} as they
usually have very competitive prediction accuracy and there is great
off-the-shelf software
\citep{R-random-forest,athey2016generalize}. A common technique to
control the remainder terms in semiparametric estimation is
cross-fitting (i.e.\ sample splitting), where the predicted value (of
the nuisance functions) for the $i$-th data point is fitted using a
subsample of the data that does not contain the $i$-th data point
\citep{schick1986asymptotically,chernozhukov2016double,athey2017efficient,newey2018cross}. For
example, we can split the data into two halves and for the $i$-th data point,
$\mu_y(\bm X_i)$ and $\mu_t(\bm X_i)$ will be estimated using the
other half of the data. This
techinque will be used in our theoretical investigation. Let the
estimated nuisance functions at the $i$-th data point be
$\hat{\mu}_y^{(-i)}(\bm{X}_i)$ and $\hat{\mu}_t^{(-i)}(\bm{X}_i)$. For notational
simplicity, we suppress the dependence on the subsample used to train these models and simply
denote them as $\hat{\mu}_y(\bm{X}_i)$ and $\hat{\mu}_t(\bm X_i)$ below.

In the second stage,
we plug in these estimates in \eqref{eq:transformed-model} and select
a model for effect modification by solving
\begin{equation}
  \label{eq:lasso}
  \hat{\bm{\beta}}_{\mathcal{M}}(\lambda) = \argmin_{\alpha,\,\bm{\beta}_{\mathcal{M}}} \sum_{i=1}^n \Big\{\big[Y_i
  - \hat{\mu}_y(\mathbf{X}_i)\big] - \big[T_i - \hat{\mu}_t(\mathbf{X}_i)\big] \cdot
  (\alpha + \mathbf{X}_{i,\mathcal{M}}^T \bm{\beta}_{\mathcal{M}})\Big\}^2 +
  \lambda \|\bm{\beta}_{\mathcal{M}}\|_1
\end{equation}
with $\mathcal{M} = \{1,\dotsc,p\}$ being the full model in this step. Let
the selected model $\hat{\mathcal{M}}$ be the nonzero entries of
$\hat{\bm{\beta}}_{\{1,\dotsc,p\}}(\lambda)$. Consider the unpenalized least
squares solution $\hat{\bm{\beta}}_{\hat{\mathcal{M}}} =
\hat{\bm{\beta}}_{\hat{\mathcal{M}}}(0)$ using the selected model
$\hat{\mathcal{M}}$, which belongs to the ``relaxed lasso'' class of estimator
$\hat{\bm{\beta}}_{\hat{\mathcal{M}}}(\lambda)$ indexed by $\lambda$
\citep{meinshausen2007relaxed}. Naturally, $\hat{\bm{\beta}}$
estimates the following (weighted) projection of $\Delta(\mathbf{x})$
to the submodel spanned by $\mathbf{X}_{\cdot,\hat{\mathcal{M}}}$,
\begin{equation} \label{eq:beta-star-causal}
\bm{\beta}_{\hat{\mathcal{M}}}^{*} =
\bm{\beta}_{\hat{\mathcal{M}}}^{*}(\mathbf{T},\mathbf{X}) =
\argmin_{\alpha,\,\bm{\beta}_{\hat{\mathcal{M}}}} \sum_{i=1}^n
\big[T_i - \mu_t(\mathbf{X}_i)\big]^2\big[\Delta(\mathbf{X}_i) - \alpha -
\mathbf{X}_{i,\hat{\mathcal{M}}}^T \bm{\beta}_{\hat{\mathcal{M}}}\big]^2.
\end{equation}
\rev{This can be interpreted as the best linear approximation of
  $\Delta(\mathbf{x})$ in the observed sample if observation $i$ is
  weighted by $\big[T_i - \mu_t(\mathbf{X}_i)\big]^2$.}
However, since the submodel $\hat{\mathcal{M}}$ is selected using the data, we must
adjust for this fact to obtain the selective distribution of
$\hat{\bm{\beta}}_{\hat{\mathcal{M}}}$. Our main theoretical
contribution in this paper is to show that the pivotal statistic obtained by
\citet{lee2013exact} is asymptotically valid under the standard rate
assumptions in semiparametric regression
\citep[e.g.][]{robinson1988root}. The main
challenge is that the estimation error in $\mu_y(\mathbf x)$ and
$\mu_t(\mathbf x)$ further complicates the model selection
event and the pivotal statistics.

For practitioners, our proposal can
be easily implemented by using $Y_i - \hat{\mu}_y(\bm X_i)$ as the
response and $(T_i - \hat{\mu}_t(\bm X_i)) \cdot \bm X_i$ as the
regressors in existing softwares for selective inference
\citep{tibshirani2017selectiveinference}.
Notice that it is not necessary to use the same set of covariates to
remove confounding (estimate $\mu_t(\bm x)$ and $\mu_y(\bm x)$) and
make selective inference for effect modification. Thus the
practitioner can simply run selective inference
using $Y_i - \hat{\mu}_y(\bm X_i)$ as the response and $(T_i -
\hat{\mu}_t(\bm X_i)) \cdot \bm Z_i$ as the regressors, where $\bm
Z_i$ are user-specified potential effect modifiers. For notational
simplicity, we will assume $\bm Z_i = \bm X_i$ from this point forward.

The rest of this paper is organized as
follows. \Cref{sec:select-infer-line} reviews the selective inference
in the linear model
\eqref{eq:illustration-model} and \Cref{sec:infer-fixed-model} reviews
the asymptotics of the semiparametric regression estimator
$\hat{\bm{\beta}}_{\mathcal{M}}(0)$ with fixed model $\mathcal{M}$ and
no regularization. \Cref{sec:select-infer-effect1} presents our main
result. \Cref{sec:simulation} verifies
the asymptotic results through simulations and studies the performance
of the selective confidence intervals in finite
sample and high dimensional settings. Readers who are not interested
in the technical details can
skip these Sections and directly go to \Cref{sec:appl-obes-elev},
where we discuss an application of the proposed method to an
epidemiological study. \Cref{sec:discussion} concludes
the paper with some further discussion.



\section{\rev{Review of} selective inference in linear models}
\label{sec:select-infer-line}

We briefly review selective inference for linear models using the
lasso in \citet{lee2013exact}. \rev{This corresponds to the simplifed
  version \eqref{eq:illustration-model} of our problem} where $\eta(\bm{x}) =
0$ for all $\bm{x} \in \mathcal{X}$ and $T \equiv 1$. First we
\rev{follow \citet{berk2013valid} and \citet{lee2013exact} and}
define the inferential target rigorously. For simplicity, we assume
$\mathbf{Y}$ and every column of $\mathbf{X}$ are centered so their
sample mean is $0$. For any submodel
$\mathcal{M} \subseteq \{1,\dotsc,p\}$, we are interested in the
parameter $\bm{\beta}^{*}_{\mathcal{M}}$ such that
$\mathbf{X}_{i,\mathcal{M}}^T\bm{\beta}_{\mathcal{M}}^{*}$ is the overall best
approximation to the true mean of $Y_i$, $\Delta(\mathbf{X}_i)$, in
the sense that
\begin{equation} \label{eq:beta-star-linear}
  \bm{\beta}_{\mathcal{M}}^{*} = \argmin_{\bm{\beta}_{\mathcal{M}} \in
    \mathbb{R}^{|\mathcal{M}|}} \sum_{i=1}^n \Big(\Delta(\mathbf{X}_i) -
  \mathbf{X}_{i,\mathcal{M}}^T\bm{\beta}_{\mathcal{M}} \Big)^2.
\end{equation}
We do not need to consider the intercept term because the data are
centered. Let \[
\mathbf{X}_{\cdot,\mathcal{M}}^{\dagger} =
(\mathbf{X}_{\cdot,\mathcal{M}}^T \mathbf{X}_{\cdot,\mathcal{M}})^{-1}
\mathbf{X}_{\cdot,\mathcal{M}}^T
\]
be the pseudo-inverse of the matrix
$\mathbf{X}_{\cdot,\mathcal{M}}$ (the submatrix of $\mathbf{X}$ with
columns in $\mathcal{M}$), so $\bm{\beta}_{\mathcal{M}}^{*} =
\mathbf{X}_{\cdot,\mathcal{M}}^{\dagger} \bm{\Delta}$ where
$\bm{\Delta} =
(\Delta(\mathbf{X}_1),\dotsc,\Delta(\mathbf{X}_n))^T$.

We are interested in making inference for
$\bm{\beta}_{\hat{\mathcal{M}}}^{*}$ where $\hat{\mathcal{M}}$
contains all the nonzero entries of the solution to the lasso problem
\eqref{eq:illustration-lasso}. Notice that
\eqref{eq:illustration-lasso} is the same as \eqref{eq:lasso} and
\eqref{eq:beta-star-linear} is the same as \eqref{eq:beta-star-causal} by
taking $\hat{\mu}_t(\bm x) = \hat{\mu}_y(\bm x) = 0$ since $T_i \equiv
1$. We assume the noise
$\epsilon_i$ is i.i.d.\ normal with variance $\sigma^2$. The normality
assumption can be relaxed in large samples
\citep{tian2017asymptotics}. A natural estimator of
$\bm{\beta}_{\hat{\mathcal{M}}}^{*}$ is the least squares solution
$\hat{\bm{\beta}}_{\hat{\mathcal{M}}} =
\mathbf{X}_{\cdot,\hat{\mathcal{M}}}^{\dagger} \mathbf{Y}$ that treats
$\hat{\mathcal{M}}$ as known. However, to obtain the sampling distribution of
$\hat{\bm{\beta}}_{\hat{\mathcal{M}}}$, the immediate challenge is
that the submodel $\hat{\mathcal{M}}$ is selected using the data,
therefore the usual normal distribution of
the least squares estimator does not hold.

To solve this problem, \citet{lee2013exact} proposed to use the
conditional distribution $\hat{\bm{\beta}}_{{\mathcal{M}}} |
\hat{\mathcal{M}} = \mathcal{M}$ to construct a pivotal statistic for
$\bm{\beta}_{\hat{\mathcal{M}}}^{*}$. Let $\hat{\mathbf{s}}$ be the sign
of the solution to the lasso problem \eqref{eq:illustration-lasso}. They found that the event
$\{\hat{\mathcal{M}} = \mathcal{M}\}$ can be written as the union of
some linear constraints on the response $\mathbf{Y}$,
\begin{equation} \label{eq:partition}
  \{\hat{\mathcal{M}} = \mathcal{M}\} = \bigcup_{\bm{s}}
  \big\{\hat{\mathcal{M}} = \mathcal{M}, \hat{\mathbf{s}} = \mathbf{s} \big\} =
  \bigcup_{\bm{s}} \big\{\mathbf{A}(\mathcal{M},\mathbf{s}) \bm{Y} \le \mathbf{b}(\mathcal{M},\mathbf{s})\big\}.
\end{equation}
The constraints are given by $\mathbf{A}(\mathcal{M},\mathbf{s}) =
(\mathbf{A}_0(\mathcal{M},\mathbf{s})^T,
\mathbf{A}_1(\mathcal{M},\mathbf{s})^T)^T$, $\mathbf{b}(\mathcal{M},\mathbf{s}) =
(\mathbf{b}_0(\mathcal{M},\mathbf{s})^T,
\mathbf{b}_1(\mathcal{M},\mathbf{s})^T)^T$, where $\mathbf{A}_0$
satisfies $\mathbf{A}_0 \mathbf{X}_{\cdot,\mathcal{M}} = \mathbf{0}$
(the exact expressions for $\mathbf{A}_0$ and $\mathbf{b}_0$ can be found in
\citet{lee2013exact}), and
\[
  \begin{split}
    &\mathbf{A}_1(\mathcal{M}, \mathbf{s}) = - \mathrm{diag}(\mathbf{s})
    \mathbf{X}_{\cdot,\mathcal{M}}^{\dagger},~\mathbf{b}_1(\mathcal{M},
    \mathbf{s}) = - \lambda \mathrm{diag}(\mathbf{s})
    (\mathbf{X}_{\cdot,\mathcal{M}}^T \mathbf{X}_{\cdot,\mathcal{M}})^{-1} \mathbf{s}.
  \end{split}
\]

Suppose we are interested in the $j$-th component of
$\bm{\beta}^{*}_{\hat{\mathcal{M}}}$. Let $\bm{\eta}_{\mathcal{M}} =
(\mathbf{X}_{\cdot,\mathcal{M}}^{\dagger})^T \mathbf{e}_j$ where $\bm
e_j$ is the unit vector for the $j$-th coordinate, so
$\big(\bm{\beta}^{*}_{{\mathcal{M}}}\big)_j =
\bm{\eta}_{\mathcal{M}}^T \bm{\Delta}$ and
$\big(\hat{\bm{\beta}}_{{\mathcal{M}}}\big)_j =
\bm{\eta}_{\mathcal{M}}^T \bm{Y}$. In a nutshell, the main result of
\citet{lee2013exact} states that
$\big(\hat{\bm{\beta}}_{{\mathcal{M}}}\big)_j | \hat{\mathcal{M}} = \mathcal{M}$
follows a truncated normal distribution. More precisely, let $F(y;
\mu,\sigma^2,l,u)$ denote the CDF of normal variable
$\mathrm{N}(\mu,\sigma^2)$ truncated to the interval $[l,u]$, that is,
\begin{equation} \label{eq:truncated-normal}
  F(y;\mu,\sigma^2,l,u) = \frac{\Phi((y-\mu)/\sigma) -
    \Phi((l-\mu)/\sigma)}{\Phi((u-\mu)/\sigma) - \Phi((l-\mu)/\sigma)}.
\end{equation}
\citet[Theorem 5.2]{lee2013exact} showed that
\begin{lemma} (Selective inference for the lasso) \label{lem:lasso-selective-inference}
If the noise
$\epsilon_i$ are i.i.d.\ $\mathrm{N}(0,\sigma^2)$, then
\begin{equation} \label{eq:lasso-pivot}
  F(\big(\hat{\bm{\beta}}_{\mathcal{M}})_j;\big(\bm{\beta}^{*}_{\mathcal{M}}\big)_j,
  \sigma^2 \bm{\eta}_{\mathcal{M}}^T
  \bm{\eta}_{\mathcal{M}},L,U)
  \,\big|\, \hat{\mathcal{M}} = \mathcal{M}, \hat{\mathbf{s}} = \mathbf{s}
  \sim \mathrm{Unif}(0,1),
\end{equation}
where
\[
  \begin{split}
    L &= L(\bm{Y};\mathcal{M},\mathbf{s}) = \bm{\eta}_{\mathcal{M}}^T \mathbf{Y} +
    \max_{(\mathbf{A} \bm{\eta})_k < 0} \frac{b_k - (\mathbf{A}
      \mathbf{Y})_k}{(\mathbf{A} \bm{\eta}_{\mathcal{M}})_k},\\
    ~U &= U(\bm{Y};\mathcal{M},\mathbf{s}) = \bm{\eta}_{\mathcal{M}}^T \mathbf{Y} +
    \max_{(\mathbf{A} \bm{\eta})_k > 0} \frac{b_k - (\mathbf{A}
      \mathbf{Y})_k}{(\mathbf{A} \bm{\eta}_{\mathcal{M}})_k}.
  \end{split}
\]
\end{lemma}
Since $\mathbf{A}_0 \mathbf{X}_{\cdot,\mathcal{M}} = \mathbf{0}$, we
have $\mathbf{A}_0 \bm{\eta}_{\mathcal{M}} = \mathbf{0}$. Therefore the interval
$[L,U]$ only depends on $\mathbf{A}_1$, which corresponds to the set
of constraints on the active variables.

To construct the selective confidence interval for
$\big(\bm{\beta}^{*}_{\hat{\mathcal{M}}}\big)_j$, one can invert the
pivotal statistic \eqref{eq:lasso-pivot} by finding values $D_j^{-}$ and
$D_j^{+}$ such that
\begin{equation} \label{eq:invert-pivot}
  F(\big(\hat{\bm{\beta}}_{\hat{\mathcal{M}}})_j;D_j^{-},
  \sigma^2 \bm{\eta}_{\hat{\mathcal{M}}}^T \bm{\eta}_{\hat{\mathcal{M}}},L,U) = 1 - q / 2, \quad
  F(\big(\hat{\bm{\beta}}_{\hat{\mathcal{M}}})_j;D_j^{+},
  \sigma^2 \bm{\eta}_{\hat{\mathcal{M}}}^T \bm{\eta}_{\hat{\mathcal{M}}},L,U) = q / 2.
\end{equation}
Then by \eqref{eq:lasso-pivot} it is easy to show that the confidence
interval $[D_j^{-},D_j^{+}]$ controls the selective type I error
\eqref{eq:stie} (if further conditioning on the event
$\{\hat{\mathbf{s}} = \mathbf{s}\}$) and hence the false coverage rate
\eqref{eq:fcr}. One can further improve the power of selective
inference by marginalizing over the coefficient signs $\mathbf{s}$, see
\citet[Section 5.2]{lee2013exact} for more detail.


\section{Inference for a fixed model of effect modification}
\label{sec:infer-fixed-model}

We now turn to the causal model \eqref{eq:causal-model} \rev{without
  the simplifying assumption that $\eta(\bm x) \equiv 0$ and $T \equiv
  1$}. \rev{As explained in \Cref{sec:our-proposal}, the submodel
  parameter $\bm \beta_{\mathcal{M}}^{*}$ is defined by the weighted
  projection \eqref{eq:beta-star-causal} instead of
  \eqref{eq:beta-star-linear}. First, we state the fundamental assumptions that are necessary for
statistical inference for the conditional average treatment effect
$\Delta(\bm x)$.}

\begin{assumption} (Fundamental assumptions in causal inference) \label{assump:basic}
For $i=1,\dotsc,n$,
\begin{enumerate}[label=(\ref{assump:basic}\Alph*)]
\item \emph{Consistency} of the observed outcome: $Y_i = Y_i(T_i)$;
\item \emph{Unconfoundedness} of the treatment assignment: $T_i \independent
  Y_i(t) | \mathbf{X}_i,~\forall t \in \mathcal{T}$;
\item \emph{Positivity (or Overlap)} of the treatment assignment:
  $T_i | \mathbf{X}_i$ has a
  positive density with respect to a
  dominating measure on $\mathcal{T}$. In particular, we assume
  $\mathrm{Var}(T_i|\mathbf{X}_i)$ exists and is between $1/C$ and $C$ for
  some constant $C > 1$ and all $\mathbf{X}_i \in \mathcal{X}$.
\end{enumerate}
\end{assumption}
Assumption (\ref{assump:basic}A) connects the observed outcome with
the potential outcomes and states that there is no interference
between the observations. Assumption (\ref{assump:basic}B) assumes
that there is no unmeasured confounding variable and is crucial to
identify the causal effect of $T$ on $Y$. This assumption is trivially
satisfied in a randomized experiment ($T_i \independent \mathbf{X}_i$).
Assumption
(\ref{assump:basic}C) ensures that statistical inference of the
treatment effect is possible. All the assumptions are essential and
commonly found in causal inference, see
\citet{rosenbaum1983central,hernan2017causal}.

In this section we consider the case of a fixed model of effect
modification, where we want to approximate $\Delta(\mathbf{X}_i)$ with
$\mathbf{X}_{i,\mathcal{M}}^T \bm{\beta}_{\mathcal{M}}$ in the sense
that it is the best linear approximation to the data generating model in
\eqref{eq:transformed-model}. Formally, the inferential target is
defined by \eqref{eq:beta-star-causal} (replacing $\hat{\mathcal{M}}$
by $\mathcal{M}$). This is slightly different from the parameter in the linear
model defined \eqref{eq:beta-star-linear} because the outcome
regression also involves the treatment variable. Similar to
\Cref{sec:infer-fixed-model}, we assume the response $Y_i -
\hat{\mu}_y(\mathbf{X}_i)$ and the design $(T_i -
\hat{\mu}_t(\mathbf{X}_i)) \bm{X}_i$, $i=1,\dotsc,n$, are all
centered, so we will ignore the intercept term in the theoretical
analysis below.

As described in \Cref{sec:introduction}, a natural estimator of
$\bm{\beta}_{\mathcal{M}}^{*}$ is the least squares estimator
$\hat{\bm{\beta}} = \hat{\bm{\beta}}_{\mathcal{M}}(0)$ defined in \eqref{eq:lasso} with
the plug-in nuisance estimates $\hat{\mu}_t(\bm{x})$ and $\hat{\mu}_y(\bm{y})$
and no regularization. The problem is: how accurate do
$\hat{\mu}_t(\bm{x})$ and $\hat{\mu}_y(\bm{y})$ need to be so that
$\hat{\bm{\beta}}_{\mathcal{M}}(0)$ is consistent and asymptotically
normal? One challenge of the theoretical analysis is that
both the regressors and the responses in \eqref{eq:lasso} involve
the estimated regression functions. Our analysis hinges on the
following modification of $\bm{\beta}^{*}_{\mathcal{M}}$:
\begin{equation} \label{eq:beta-tilde}
  \tilde{\bm{\beta}}_{\mathcal{M}}(\mathbf{T},\mathbf{X}) = \argmin_{\bm{\beta}_{\mathcal{M}} \in
    \mathbb{R}^{|\mathcal{M}|}} \frac{1}{n} \sum_{i=1}^n \big( T_i - \hat{\mu}_t(\mathbf{X}_i) \big)^2 \big(\Delta(\mathbf{X}_i) - \mathbf{X}_{i,\mathcal{M}}^T
  \bm{\beta}_{\mathcal{M}}\big)^2.
\end{equation}
We use tilde in this paper to indicate that the
quantity also depends on $\hat{\mu}_t(\cdot)$ and/or
$\hat{\mu}_y(\cdot)$.
The next Lemma shows that $\tilde{\bm{\beta}}_{\mathcal{M}}$ is
very close to the target parameter $\bm{\beta}^{*}_{\mathcal{M}}$ when
the treatment model is sufficiently accurate.

\begin{assumption} (Accuracy of treatment
  model) \label{assump:accuracy-treatment}
  $\|\hat{\mu}_t - \mu_t\|_2^2 = (1/n)\sum_{i=1}^n
    (\hat{\mu}_t(X_i) - \mu_t(X_i))^2 = o_p(n^{-1/2})$.
\end{assumption}

\begin{assumption} \label{assump:support-X}
  The support of $\mathbf{X}$ is uniformly bounded, i.e.\
  $\mathcal{X} \subseteq [-C,C]^p$ for some constant $C$. The conditional
  treatment effect $\Delta(\mathbf{X})$ is also bounded by $C$.
\end{assumption}

\begin{lemma} \label{thm:hat-star}
  Suppose
  \Cref{assump:basic,assump:accuracy-treatment,assump:support-X} are
  satisfied. For a fixed model $\mathcal{M}$ such that
  $\mathrm{E}[\bm{X}_{i,\mathcal{M}} \bm{X}_{i,\mathcal{M}}^T] \succeq (1/C)
  \mathbf{I}_{|\mathcal{M}|}$, we have
  $\|\tilde{\bm{\beta}}_{\mathcal{M}} -
  {\bm{\beta}}^{*}_{\mathcal{M}}\|_{\infty} = o_p(n^{-1/2})$.
\end{lemma}

\begin{assumption} (Accuracy of outcome model) \label{assump:accuracy-outcome}
  $\|\hat{\mu}_y - \mu_y\|_2 = o_p(1)$ and
  $\|\hat{\mu}_t - \mu_t\|_2 \cdot \|\hat{\mu}_y - \mu_y\|_2 =
  o_p(n^{-1/2})$.
\end{assumption}

The last assumption is a doubly-robust type assumption
\citep{chernozhukov2016double}. Double
robustness is the property that the causal effect estimate is consistent
if at least one of the two nuisance estimates are consistent
\citep{bang2005doubly}.  To make asymptotic normal
inference, we additionally require both nuisance estimates to be consistent and converge fast enough: \Cref{assump:accuracy-treatment,assump:accuracy-outcome}
are satisfied if $\hat{\mu}_y$ is consistent and both $\hat{\mu}_t$ and $\hat{\mu}_y$ converge faster
than the familiar $n^{-1/4}$ rate in $l_2$ norm
\citep{van2017generally,chernozhukov2016double}.

The next Theorem
establishes the asymptotic distribution of
$\hat{\bm{\beta}}_{\mathcal{M}}$ when $\mathcal{M}$ is fixed and given.

\begin{theorem} \label{thm:beta-fixed}
  Suppose cross-fitting is used to estimate the nuisance functions. Under
  \Cref{assump:basic,assump:accuracy-treatment,assump:support-X,assump:accuracy-outcome}, for a fixed model $\mathcal{M}$ such that
  $\mathrm{E}[\bm{X}_{i,\mathcal{M}} \bm{X}_{i,\mathcal{M}}^T] \succeq (1/C)
  \mathbf{I}_{|\mathcal{M}|}$, we have
  \[
  \Big(\sum_{i=1}^n (T_i - \hat{\mu}_t(\mathbf{X}_{i}))^2
  \mathbf{X}_{i,\mathcal{M}} \mathbf{X}_{i, \mathcal{M}}^T\Big)^{-1/2} (\hat{\bm{\beta}}_{\mathcal{M}} -
  {\bm{\beta}}^{*}_{\mathcal{M}}) \overset{d}{\to} \mathrm{N}(0,\sigma^2 \mathbf{I}_{|\mathcal{M}|}).
\]
\end{theorem}

The key step to prove this Theorem is to replace
$\bm{\beta}^{*}_{\mathcal{M}}$ by $\tilde{\bm{\beta}}_{\mathcal{M}}$
using \Cref{thm:hat-star}. The rest of the proof is just an
extension to the standard asymptotic analysis of least squares
estimator in which the response is perturbed
slightly. \rev{\Cref{thm:beta-fixed} is an instance of the finite
  population central limit theorem in which the estimand $\bm
  \beta^{*}_{\mathcal{M}}$ depends on the sample
  \citep{li2017general}.}

We want to emphasize that in a randomized experiment (i.e.\ $T \independent
\mathbf{X}$) or the parametric form of $\mu_t(\mathbf{x})$ is known,
$\mu_t$ can be estimated very accurately by the sample
ratio or a parametric regression. In this case, $\|\hat{\mu}_t - \mu_t\|_2 =
O_p(n^{-1/2})$ so \Cref{assump:accuracy-treatment}
holds trivially. \Cref{assump:accuracy-outcome} is reduced
to the very weak condition that $\hat{\mu}_y(\mathbf{x})$ is
consistent. This is easily satisfied by standard nonparametric
regressions or the random forests
\citep{biau2012analysis,scornet2015consistency}.


\section{Selective inference for effect modificiation}
\label{sec:select-infer-effect1}

As argued in \Cref{sec:introduction}, it is often desirable to
use a simple  model to approximately describe the effect modification when the
dimension of
$\bm{X}$ is high. One way to do this is to solve the lasso problem
\eqref{eq:lasso} and let the selected model $\hat{M} =
\hat{M}_{\lambda}$ be the positions of the non-zero entries in the solution
$\hat{\bm{\beta}}_{\{1,\dotsc,p\}}(\lambda)$. We want to make
valid inference for the parameter
$\bm{\beta}^{*}_{\hat{\mathcal{M}}}$ defined in
\eqref{eq:beta-star-causal} given the fact that $\hat{\mathcal{M}}$ is
selected using the data.

Compared to the selective inference in linear models described in
\Cref{sec:select-infer-line}, the challenge here is that the nuisance
parameters $\mu_y(\bm{x})$ and $\mu_t(\bm{x})$ must be estimated by
the data. This means that in the regression model
\eqref{eq:transformed-model}, the response $Y_i - \mu_y(\mathbf{X}_i)$
and the regressors (in the approximate linear model) $(T_i -
\mu_t(\mathbf{X}_i)) \mathbf{X}_i$ are not observed exactly. Similar
to \Cref{sec:infer-fixed-model}, the
estimation error $\|\hat{\mu}_t - \mu_t\|_2$ and $\|\hat{\mu}_y
- \mu_y\|_2$ must be sufficiently small to make the asymptotic
theory go through. Our main technical result is that with some
additional assumptions on the selection event, the same
rate assumptions in the fixed model case
(\Cref{assump:accuracy-treatment,assump:accuracy-outcome})
ensures that
the pivotal statistic \eqref{eq:lasso-pivot} is still asymptotically valid.

The first key assumption we make is that the size of the select model
$\hat{\mathcal{M}}$ is not too large. This assumption is important to
control the number of parametric models we need to consider in the
asymptotic analysis.
\begin{assumption} (Size of the selected model) \label{assump:model-size}
  For some constant $m$, $\mathrm{P}(|\hat{\mathcal{M}}| \le m) \to
  1$.
\end{assumption}
Similar to \Cref{thm:hat-star}, we assume the covariance
matrices of the design $\mathbf{X}$ are uniformly positive definite,
so the regressors are not collinear in any selected model.
\begin{assumption} (Sparse eigenvalue assumption) \label{assump:design}
  For all model $\mathcal{M}$ such that $|\mathcal{M}| \le m$,
  $\mathrm{E}[\bm{X}_{i,\mathcal{M}} \bm{X}_{i,\mathcal{M}}^T] \succeq (1/C)
  \mathbf{I}_{|\mathcal{M}|}$.
\end{assumption}

These additional assumptions ensure the modified parameter
$\tilde{\bm{\beta}}_{\hat{\mathcal{M}}}$ is not too far from the
target parameter ${\bm{\beta}}^{*}_{\hat{\mathcal{M}}}$ when the
treatment model is sufficiently accurate.
\begin{lemma} \label{lem:beta-tilde-star-random}
  Under the assumptions in \Cref{thm:hat-star} and additionally
  \Cref{assump:model-size,assump:design},
  $\|{\bm{\beta}}^{*}_{\hat{\mathcal{M}}}\|_{\infty} = O_p(1)$ and
  $\|\tilde{\bm{\beta}}_{\hat{\mathcal{M}}} -
  {\bm{\beta}}^{*}_{\hat{\mathcal{M}}}\|_{\infty} =
  o_p(n^{-1/2})$.
\end{lemma}

Let $\tilde{\mathbf{X}}_{i,\mathcal{M}} = (T_i - \hat{\mu}_t(\mathbf{X}_i)) \mathbf{X}_{i,\mathcal{M}}$ be the
transformed design (tilde indicates dependence on $\hat{\mu}_t$) and $\tilde{\bm{\eta}}_{\mathcal{M}} = (\tilde{\mathbf{X}}_{\cdot,\mathcal{M}}^{\dagger})^T
  \mathbf{e}_j$ be the linear transformation we are interested in. In
  other words, $(\hat{\bm{\beta}}_{\hat{\mathcal{M}}})_j =
  \tilde{\bm{\eta}}^T_{\mathcal{M}} \tilde{\bm Y}$
  where $\tilde{\bm Y} = \mathbf{Y} - \hat{\bm{\mu}}_y$ and
  $\hat{\bm{\mu}}_y$ is the vector of fitted values of the
  $\mathbf{Y}$ versus $\mathbf{X}$ regression, $\hat{\bm{\mu}}_y =
  (\hat{\mu}_y(\mathbf{X}_1), \dotsc,
  \hat{\mu}_y(\mathbf{X}_n))^T$. The model selection event \eqref{eq:partition}
  can be obtained analogously by substituting $\bm X$ with the estimated
  transformed design $\tilde{\mathbf{X}}$ in the definition of the matrices
  $\bm{A}(\mathcal{M}, \bm{s})$ and vectors $\bm{b}(\mathcal{M},\bm{s})$.

Next we state the extra technical assumptions for our main Theorem.
\begin{assumption} (Truncation threshold) \label{assump:smooth-pivot}
  The truncation thresholds $L$ and $U$ (computed using the
  transformed design $\tilde{\bm{X}}$) satisfy
  \[
  \mathrm{P}\Big(\frac{U(\mathbf{Y} - \hat{\bm{\mu}}_y) - L(\mathbf{Y}
    - \hat{\bm{\mu}}_y)}{\sigma\|\tilde{\bm{\eta}}_{\mathcal{M}}\|} \ge 1/C\Big) \to
  1.
  \]
\end{assumption}

\begin{assumption} (Lasso solution) \label{assump:kkt}
$\mathrm{P}\Big(\big|\big(\hat{\bm{\beta}}_{\{1,\dotsc,p\}}(\lambda)\big)_k\big| \ge
1/(C\sqrt{n}), ~\forall k \in \hat{\mathcal{M}}\Big) \to 1$.
\end{assumption}

\Cref{assump:smooth-pivot} assumes the truncation points $L$ and $U$
are not too close (i.e.\ the conditioning event is not too small),
so a small perturbation does not change the denominator of
\eqref{eq:truncated-normal} a lot. \Cref{assump:kkt} assumes the
lasso solution does not have a small coefficient. This is true with
high probability if the truth is a sparse linear model and the true nonzero
coefficients are not too small; see
\citet{negahban2012unified}. However, \Cref{assump:kkt} does not
require the true model is sparse and must be selected
consistently. Together these two assumptions mean the selected model
is not a small probability event and is stable to small perturbations. Notice that
both these assumptions can be verified empirically.

Finally we state our main Theorem. Note that we assume the noise is
homoskedastic and Gaussian in this Theorem, but it is possible to
relax this assumption. See \Cref{sec:assumptions-paper} for more
discussion about all the assumptions in this paper.
\begin{theorem} \label{thm:asymptotic-validity}
  Suppose cross-fitting is used to estimate the nuisance functions
  and the noise $\epsilon_i$ are i.i.d.\ $\mathrm{N}(0,\sigma^2)$. Under
  \Cref{assump:basic,assump:accuracy-treatment,assump:support-X,assump:accuracy-outcome,assump:model-size,assump:design,assump:smooth-pivot,assump:kkt},
  the pivotal statistic in \eqref{eq:lasso-pivot} is asymptotically
  valid. More specifically, for any $\mathcal{M}$ such that
  $\mathrm{P}(\hat{\mathcal{M}} = \mathcal{M},\hat{\mathbf{s}} = \mathbf{s}) > 0$,
  \begin{equation} \label{eq:asymptotic-validity}
  F\Big(\big(\hat{\bm{\beta}}_{{\mathcal{M}}}\big)_j;
  \big(\bm{\beta}^{*}_{{\mathcal{M}}}\big)_j, \sigma^2 \tilde{\bm{\eta}}_{\mathcal{M}}^T
  \tilde{\bm{\eta}}_{\mathcal{M}},L\big(\mathbf{Y} - \hat{\bm{\mu}}_y;{\mathcal{M}},{\mathbf{s}}\big),U\big(\mathbf{Y} - \hat{\bm{\mu}}_y;{\mathcal{M}},{\mathbf{s}}\big)
  \Big) \Big| \hat{\mathcal{M}} = \mathcal{M}, \hat{\mathbf{s}} =
  \mathbf{s} \overset{d}{\to} \mathrm{Unif}(0,1).
  \end{equation}
\end{theorem}

The main challenge in proving \Cref{thm:asymptotic-validity} is that the
CDF $F$ is a ratio (defined in \Cref{eq:truncated-normal}), so it is
necessary to bound the error in both the numerator and the
denominator. Also, the truncation limits $L$ and $U$ involve taking
maxima over many
constraints, which lead to additional compications. Notice that our proof of \Cref{thm:asymptotic-validity}
can be easily extended to other variable selection methods such as
forward stepwise regression as long as the selection
event $\{\hat{\mathcal{M}} = \mathcal{M}\}$ can be characterized as
linear constraints of $\mathbf{Y}$
\citep{loftus2014significance,taylor2015statistical}. In this case,
\Cref{assump:kkt}
needs to be replaced by the condition that these linear constraints
are satisfied with at least $O(1/\sqrt{n})$ margin. See \Cref{lem:kkt}
in the Appendix.

Similar to the case in \Cref{sec:infer-fixed-model}, the pivot in
\eqref{eq:asymptotic-validity} has no unknown parameter (except for
$\sigma^2$, which is assumed to be known in the selective inference of
\citet{lee2013exact}) and can be inverted as in
\eqref{eq:invert-pivot} to obtain the confidence intervals for the
coefficients $\bm{\beta}^{*}_{\hat{\mathcal{M}}}$.


\section{Simulation}
\label{sec:simulation}

\subsection{Validity of selective inference for effect modification}
\label{sec:select-infer-effect}

We evaluate the method proposed in this paper with data simulated
from the causal model \eqref{eq:causal-model}. We consider a
comprehensive simulation design parametrized by the following
parameters
\begin{itemize}
\item $s_t$: sparsity of $\mu_t$, either $0$ (a
  randomized experiment), $5$, or $25$.
\item $f_t$: functional form of $\mu_t$, either linear (lin), quadratic
  (qua), a five-variate function used by \citet{friedman1989flexible}
  (FS), or a five-variate function used by
  \citet*{friedman1983multidimensional} (FGS); see below for detail.
\item $s_y$: sparsity of $\mu_y$, either $5$ or $25$.
\item $f_y$: functional form of $\mu_y$, same options as $f_t$.
\item $s_{\Delta}$: sparsity of $\Delta$, either $5$ or $25$.
\item $f_{\Delta}$: functional form of $\Delta$, same options as $f_t$.
\item $\sigma$: standard deviation of the noise, either
  $0.25$ or $0.5$.
\item noise: distribution of the noise, either $\sigma \cdot \mathrm{N}(0,
  1)$ or $\sigma \cdot \mathrm{double\textrm{-}exp}(0, 1/\sqrt{2})$.
\end{itemize}
These give us $3072$ simulation settings in total. The functional forms
are
\begin{itemize}
\item Linear: $f(x_1,x_2,x_3,x_4,x_5) = 3x_1 + x_2 + x_3 + x_4 + x_5 -
  3.5$;
\item Quadratic: $f(x_1,x_2,x_3,x_4,x_5) = 3(x_1 - 0.5)^2 + (x_2 -
  0.5)^2 + (x_3 - 0.5)^2 + (x_4 - 0.5)^2 + (x_5 - 0.5)^2 + 3x_1 + x_2
  + x_3 + x_4 + x_5 - 4$;
\item FS: $f(x_1,x_2,x_3,x_4,x_5) = [0.1 \exp^{4x_1} +
  4/(1+\exp^{-20(x_2-0.5)}) + 3 x_3 + 2x_4 + x_5 - 6.3]/2.5$;
\item FGS: $f(x_1,x_2,x_3,x_4,x_5) = [10 \sin(\pi x_1 x_2) + 20
  (x_3-0.5)^2 + 10 x_4 + 5 x_5 - 14.3]/4.9$.
\end{itemize}
In every setting,
we generate $n = 1000$ observations and $p = 25$ covariates that are
uniformly distributed over $[0,1]$ and independent.
If the sparsity is $5$, for example $s_t = 5$, then $\mu_t(\mathbf{x})
= f(x_1,x_2,x_3,x_4,x_5)$. If the sparsity is $25$, then $\mu_t(\mathbf{x})
= f(x_1,x_2,x_3,x_4,x_5)/1^2 + f(x_6,x_7,x_8,x_9,x_{10})/2^2 + \cdots +
f(x_{21},x_{22},x_{23},x_{24},x_{25})/5^2$ and similarly for
$\mu_y(\mathbf{x})$ and $\Delta(\mathbf{x})$. To evaluate the performance
of selective inference in high-dimensional settings, we also
simulate a more challenging setting with $p=500$ covariates by appending
$475$ independent covariates to $\bm X$.

After the data are generated, we use the lasso (with cross validation) or random forest to estimate the
nuisance functions
$\mu_t(\mathbf{x})$ and $\mu_y(\mathbf{x})$. For the lasso we use the
\texttt{R} package \texttt{glmnet}
\citep{friedman2010regularization}. For the random forest we use the
\texttt{R}
package \texttt{randomForest} \citep{R-random-forest} with all the
default tuning parameters
except \texttt{nodesize = 20} \rev{ and \texttt{mtry = 25}}. \rev{For
  $\hat{\mu}_t(\mathbf{x})$ and $\hat{\mu}_y(\mathbf{x})$, we use the
  out-of-bag (OOB) predictions from the random forests. This may serve
as a proxy to cross-fitting.} We have also tried to use cross-fitting to
estimate the nuisance functions, but that appears to deteriorate the
performance of our methods. Thus in all the empirical investigations below, we decide
to use the full sample to estimate the nuisance functions. See
\Cref{sec:assumptions-paper} for more discussion on cross-fitting.

We select effect
modifiers using the lasso regression \eqref{eq:lasso} with $\lambda =
1.1 \times \mathrm{E}[\|\mathbf{X}\bm{\epsilon}\|_{\infty}]$ where
$\bm{\epsilon} \sim \mathrm{N}(\mathbf{0},\hat{\sigma}^2\mathbf{I}_p)$
as recommended by \citet{negahban2012unified}. The noise variance
$\sigma^2$ is estimated by the full linear regression of
$Y_i - \hat{\mu}_y(\mathbf{X}_i)$ on $[T_i -
\hat{\mu}_t(\mathbf{X}_i)]\mathbf{X}_i$, $i=1,\dotsc,n$. Finally we use
the asymptotic pivot in \eqref{eq:asymptotic-validity} to construct
selective 95\%-confidence intervals for the selected submodel as
implemented in the function \texttt{fixedLassoInf} in the \texttt{R}
package \texttt{selectiveInference} \citep{R-selective}. In
each simulation setting, we run the above procedure for $300$
independent realizations. Three error metrics are reported: the false
coverage rate (FCR) defined in \eqref{eq:fcr}, the selective type I error
(STIE) defined in \eqref{eq:stie}, and the false sign rate (FSR)
\[
\mathrm{FSR} = \mathrm{E}\bigg[\frac{\#\big\{j \in \hat{\mathcal{M}}:
    0\not\in
    [D^{-}_j,D^{+}_j],\,\big(\bm{\beta}^{*}_{\mathcal{M}}\big)_j \cdot
    D^{-}_j < 0\big\}}{\max(|\hat{\mathcal{M}}|,1)} \bigg]
\]
to examine if any significant selective confidence interval has the
incorrect sign.

\subsubsection{Results in the low-dimensional settings}
\label{sec:low-dimens-sett}

In \Cref{tab:linear} we report the simulation results in the
low-dimensional settings when the true
functional forms are all linear and the nuisance functions are estimated
by the random forest. The size of the selected model
$|\hat{\mathcal{M}}|$ seems to heavily depend on the intrinsic complexity
of the nuisance parameter ($s_t$, $s_y$) and the noise level
($\sigma$). The selective type I error and the false
coverage rate were controlled at the nominal 5\% level even when the
noise is non-Gaussian, and no significant confidence interval
containing only
incorrect signs was found. Similar conclusions can be reached from
\Cref{tab:nonlinear} where exactly one of the true functional forms is
nonlinear, with the exception that in two simulation settings the
false coverage rates (FCR) were greater than $10\%$. In both cases, the true
propensity score $\mu_t(\mathbf{x})$ is generated by the FGS and the
biases of the estimated propensity score $\hat{\mu}_t$ were larger than
those in the other settings.

\begin{table}[t]
\centering
\begin{tabular}{rlrlrlrlrrrrrr}
  \hline
$s_t$ & $f_t$ & $s_y$ & $f_y$ & $s_{\Delta}$ & $f_{\Delta}$ & $\sigma$
& noise & $|\hat{\mathcal{M}}|$ & \# sig & FCR & STIE & FSR &
$\mathrm{bias}(\hat{\mu}_t)$ \\
  \hline
0 & lin & 5 & lin & 5 & lin & 0.25 & normal & 4.18 & 3.32 & 0.051 & 0.052 & 0.000 & 0.0018 \\
  0 & lin & 5 & lin & 5 & lin & 0.5 & normal & 1.96 & 1.36 & 0.050 & 0.049 & 0.000 & -0.0015 \\
  0 & lin & 5 & lin & 5 & lin & 0.25 & exp & 4.14 & 3.19 & 0.053 & 0.056 & 0.000 & -0.0023 \\
  0 & lin & 5 & lin & 5 & lin & 0.5 & exp & 1.87 & 1.36 & 0.058 & 0.066 & 0.000 & -0.0010 \\
  5 & lin & 5 & lin & 5 & lin & 0.25 & normal & 2.02 & 1.37 & 0.021 & 0.026 & 0.000 & 0.0011 \\
  5 & lin & 5 & lin & 5 & lin & 0.5 & normal & 1.11 & 1.03 & 0.043 & 0.045 & 0.000 & 0.0021 \\
  25 & lin & 5 & lin & 5 & lin & 0.25 & normal & 1.83 & 1.37 & 0.039 & 0.038 & 0.000 & 0.0019 \\
  25 & lin & 5 & lin & 5 & lin & 0.5 & normal & 1.13 & 1.04 & 0.030 & 0.033 & 0.000 & 0.0027 \\
  0 & lin & 25 & lin & 5 & lin & 0.25 & normal & 3.23 & 2.23 & 0.044 & 0.049 & 0.000 & -0.0002 \\
  25 & lin & 5 & lin & 5 & lin & 0.25 & normal & 1.83 & 1.37 & 0.039 & 0.038 & 0.000 & 0.0019 \\
  0 & lin & 5 & lin & 25 & lin & 0.25 & normal & 4.32 & 3.36 & 0.044 & 0.045 & 0.000 & -0.0000 \\
  25 & lin & 25 & lin & 25 & lin & 0.25 & normal & 1.33 & 1.09 & 0.030 & 0.030 & 0.000 & 0.0027 \\
   \hline
\end{tabular}
\caption{Performance of the selective confidence intervals in
  low-dimensional settings where the true $\mu_t(\mathbf{x})$,
  $\mu_y(\mathbf{x})$, and $\Delta(\mathbf{x})$ are linear in
  $\mathbf{x}$ and the nuisance functions are estimated by the random
  forest. The false coverage rates (FCR) and selective type I
  error (STIE) are all close to the
  nominal $5\%$ level. Columns in this table are: sparsity of $\mu_t$
  ($s_t$), functional form of $\mu_t$ ($f_t$), sparsity of $\mu_y$
  ($s_y$), functional form of $\mu_y$ ($f_y$), sparsity of $\Delta$
  ($s_{\Delta}$), functional form of $\Delta$ ($f_{\Delta}$), standard
  deviation of the noise ($\sigma$), distribution of the noise
  (noise), average size of selected models ($|\hat{\mathcal{M}}|$),
  average number of significant partial regression coefficients (\#
  sig), false coverage rate (FCR), selective type I error (STIE), false sign rate (FSR), average bias
  of the estimated propensity score ($\mathrm{bias}(\hat{\mu}_t)$).}
\label{tab:linear}
\end{table}

\begin{table}[t]
\centering
\begin{tabular}{rlrlrlrlrrrrrr}
  \hline
$s_t$ & $f_t$ & $s_y$ & $f_y$ & $s_{\Delta}$ & $f_{\Delta}$ & $\sigma$
& noise & $|\hat{\mathcal{M}}|$ & \# sig & FCR & STIE & FSR &
$\mathrm{bias}(\hat{\mu}_t)$ \\
  \hline
0 & lin & 5 & quad & 5 & lin & 0.25 & normal & 4.25 & 3.51 & 0.059 & 0.063 & 0.000 & -0.0000 \\
  0 & lin & 5 & FS & 5 & lin & 0.25 & normal & 4.72 & 4.21 & 0.048 & 0.048 & 0.000 & -0.0009 \\
  0 & lin & 5 & FGS & 5 & lin & 0.25 & normal & 3.18 & 2.18 & 0.066 & 0.064 & 0.000 & -0.0000 \\
  0 & lin & 5 & lin & 5 & quad & 0.25 & normal & 4.08 & 3.20 & 0.058 & 0.060 & 0.000 & -0.0021 \\
  0 & lin & 5 & lin & 5 & FS & 0.25 & normal & 3.28 & 2.98 & 0.053 & 0.054 & 0.000 & 0.0007 \\
  0 & lin & 5 & lin & 5 & FGS & 0.25 & normal & 3.75 & 3.47 & 0.040 &
  0.040 & 0.000 & -0.0006 \\
  5 & quad & 5 & lin & 5 & lin & 0.25 & normal & 2.47 & 1.72 & 0.042 & 0.045 & 0.000 & -0.0030 \\
  5 & FS & 5 & lin & 5 & lin & 0.25 & normal & 2.34 & 1.61 & 0.064 & 0.060 & 0.000 & 0.0011 \\
  5 & FGS & 5 & lin & 5 & lin & 0.25 & normal & 2.13 & 1.67 & {\bf 0.136} & {\bf 0.125} & 0.000 & 0.0070 \\
  5 & lin & 5 & quad & 5 & lin & 0.25 & normal & 2.29 & 1.51 & 0.038 & 0.045 & 0.000 & 0.0023 \\
  5 & lin & 5 & FS & 5 & lin & 0.25 & normal & 2.71 & 1.89 & 0.036 & 0.034 & 0.000 & 0.0030 \\
  5 & lin & 5 & FGS & 5 & lin & 0.25 & normal & 1.44 & 1.12 & 0.083 & 0.084 & 0.000 & 0.0016 \\
  5 & lin & 5 & lin & 5 & quad & 0.25 & normal & 1.79 & 1.33 & 0.023 & 0.024 & 0.000 & 0.0014 \\
  5 & lin & 5 & lin & 5 & FS & 0.25 & normal & 2.53 & 2.19 & 0.038 & 0.036 & 0.000 & 0.0032 \\
  5 & lin & 5 & lin & 5 & FGS & 0.25 & normal & 2.82 & 2.37 & 0.032 &
  0.033 & 0.000 & 0.0008 \\
  5 & quad & 5 & lin & 5 & lin & 0.25 & exp & 2.44 & 1.68 & 0.047 & 0.051 & 0.000 & -0.0016 \\
  5 & FS & 5 & lin & 5 & lin & 0.25 & exp & 2.28 & 1.59 & 0.049 & 0.058 & 0.000 & 0.0005 \\
  5 & FGS & 5 & lin & 5 & lin & 0.25 & exp & 2.15 & 1.74 & {\bf 0.117} & 0.099 & 0.000 & 0.0098 \\
  5 & lin & 5 & quad & 5 & lin & 0.25 & exp & 2.16 & 1.56 & 0.030 & 0.032 & 0.000 & 0.0001 \\
  5 & lin & 5 & FS & 5 & lin & 0.25 & exp & 2.72 & 1.95 & 0.021 & 0.023 & 0.000 & 0.0036 \\
  5 & lin & 5 & FGS & 5 & lin & 0.25 & exp & 1.45 & 1.15 & 0.061 & 0.060 & 0.000 & 0.0006 \\
  5 & lin & 5 & lin & 5 & quad & 0.25 & exp & 1.81 & 1.35 & 0.028 & 0.028 & 0.000 & 0.0028 \\
  5 & lin & 5 & lin & 5 & FS & 0.25 & exp & 2.61 & 2.29 & 0.040 & 0.035 & 0.000 & 0.0005 \\
  5 & lin & 5 & lin & 5 & FGS & 0.25 & exp & 2.89 & 2.44 & 0.033 & 0.033 & 0.000 & 0.0024 \\
   \hline
\end{tabular}
\caption{Performance of the selective confidence intervals in
  low-dimensional settings where one of the true $\mu_t(\mathbf{x})$,
  $\mu_y(\mathbf{x})$, and $\Delta(\mathbf{x})$ is nonlinear in
  $\mathbf{x}$ and the nuisance functions are estimated by the random forest. The false coverage rates (FCR) and selective type I
  error (STIE) are close to the
  nominal $5\%$ level in almost all settings (exceptions are bolded). See caption of
  \Cref{tab:linear} for meaning of the columns.}
\label{tab:nonlinear}
\end{table}

\begin{figure}[t]
  \centering
  \includegraphics[width = 0.9\textwidth]{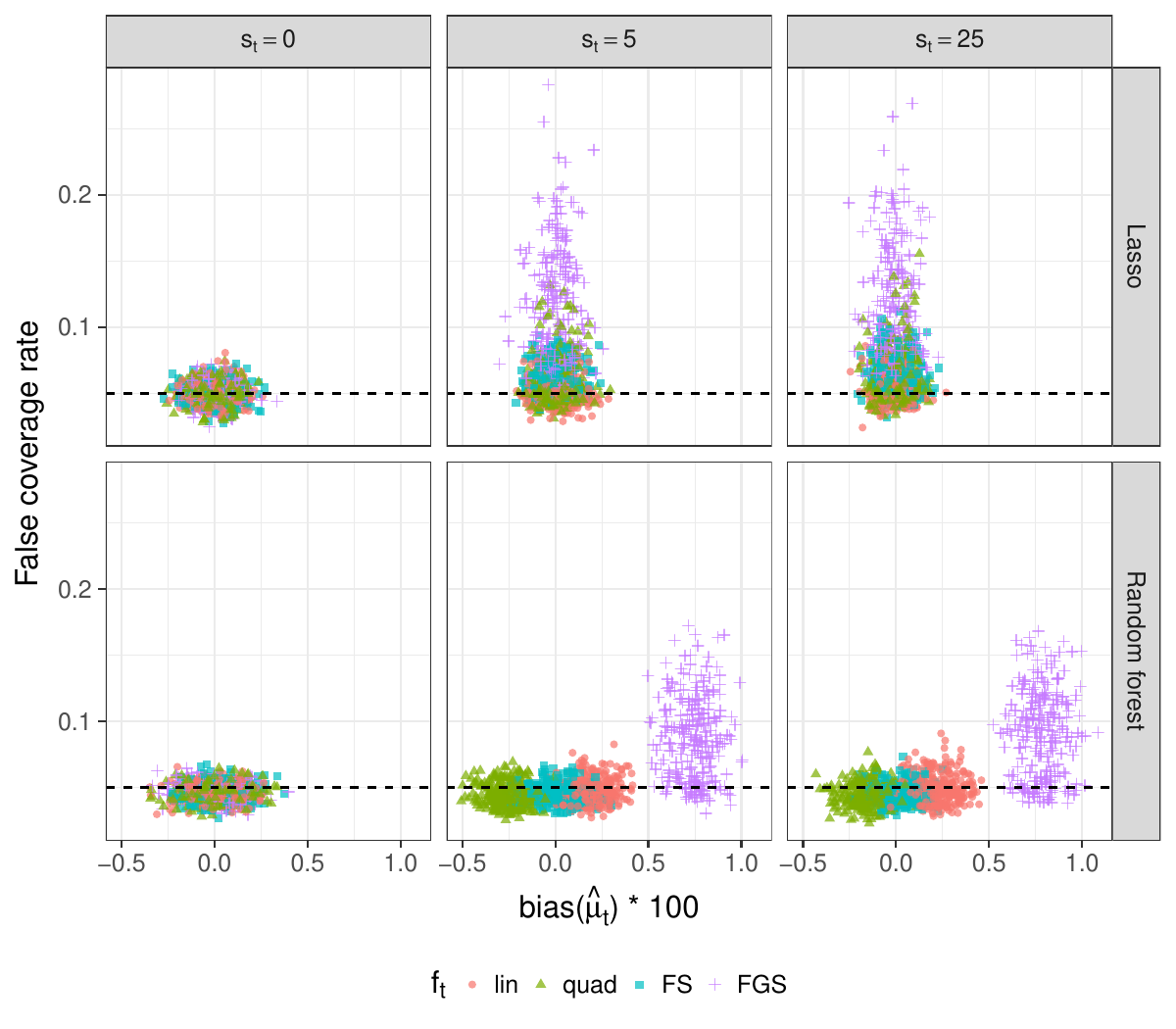}
  \caption{False coverage rate (FCR) versus bias of $\hat{\mu}_t$ in the 3072
    low-dimensional simulation settings ($p = 25$). When $s_t = 0$ (randomized experiment), the
    false coverage rate is controlled under all settings. When $s_t >
    0$ (observational data), there is no guarantee to control the
    FCR. When the nuisance functions are estimated by
    random forest in this case, the FCR is not much larger than the nominal level unless $f_t$ is
    FGS (purple crosses in the figure) where the random forest
    estimates of $f_t$ are biased.}
  \label{fig:fcr-bias-low}
\end{figure}

To get a broader picture of the performance of selective inference,
\Cref{fig:fcr-bias-low} shows the FCR versus the average
bias of $\hat{\mu}_t$ for all the $3072$ simulation settings when $\bm
X$ is low-dimensional ($p = 25$). When
$s_t = 0$ (randomized experiments), the error rates were well controlled at
the nominal 5\% level across all settings, regardless of the dimension
of $\bm X$. When $s_t > 0$ (observational studies), the rate assumption
for $\hat{\mu}_t$ (\Cref{assump:accuracy-treatment}) could be violated and
there is no
guarantee that the selective inference is still asymptotically
valid. Somewhat surprisingly, the false coverage rates were not too
high in most simulation settings. This is especially true when the
nuisance functions are estimated by the random
forest. The FCR was well controlled except when $f_t$ is
FGS, the case that the random forest estimator of  $\mu_t$ was clearly
biased. The selective inference performed poorly when $\mu_t$ and
$\mu_y$ are estimated by lasso, which is not too surprising because
some of the functional forms we used are highly nonlinear.

\subsubsection{Results in the high-dimensional settings}
\label{sec:results-high-dimens}

The coverage of selective CI deteriorated in observational studies
when the dimension of $\bm X$ is high ($p = 500$). \Cref{fig:fcr-bias-high} is the
counterpart of \Cref{fig:fcr-bias-low} in the high dimensional
settings. The FCR is usually much higher than the nominal level,
though the performance of lasso is better than the random forest. It
seems that the ``bet on sparsity'' principle
\citep{hastie2009elements} pays off to some extent here. When both
$\mu_t$ and $\mu_y$ are linear or quadratic functions and when they
are estimated by the lasso, the FCR were never larger than 10\% (this
cannot be directly observed from \Cref{fig:fcr-bias-high}).

\begin{figure}[t]
  \centering
  \includegraphics[width = 0.9\textwidth]{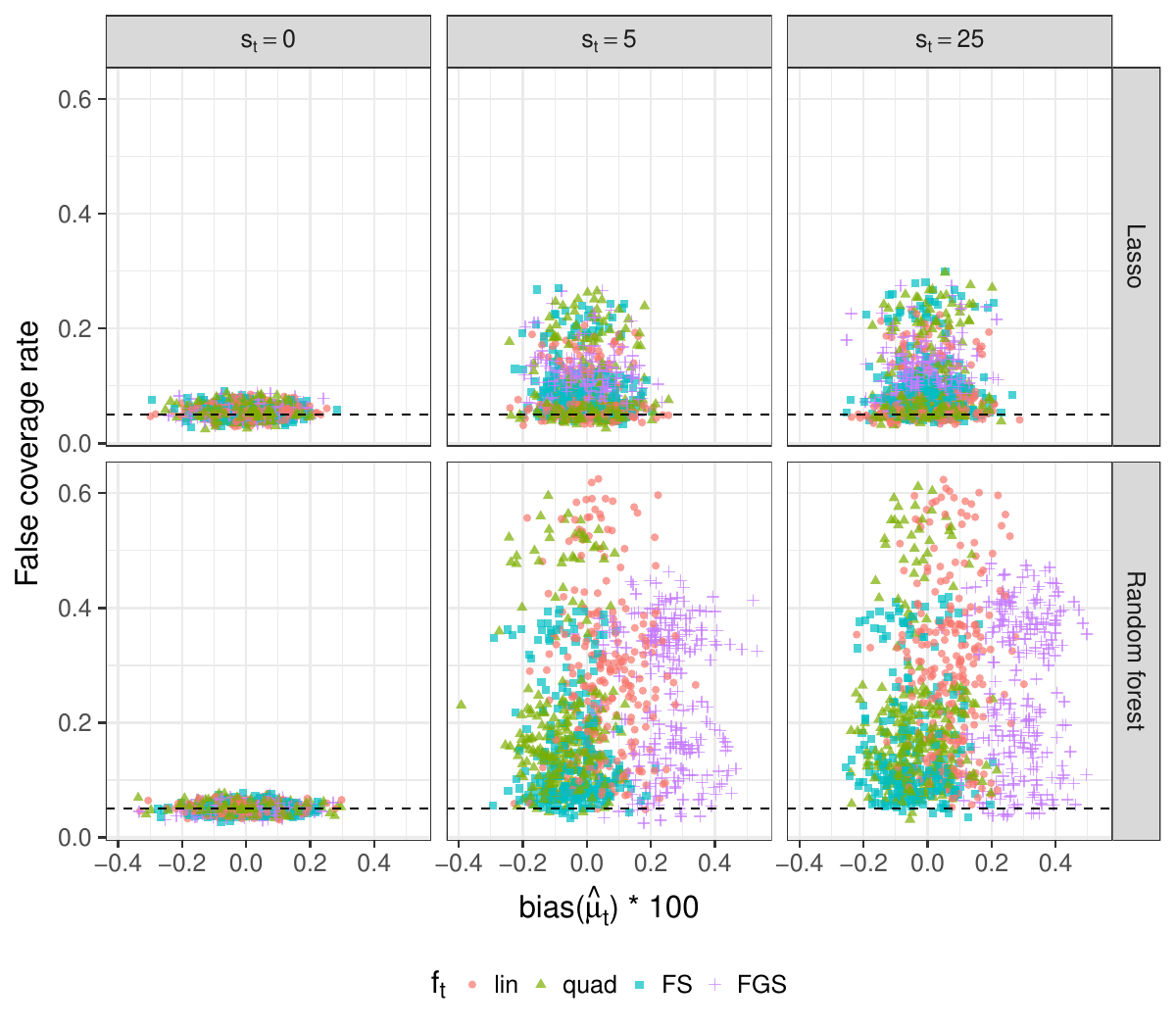}
  \caption{False coverage rate (FCR) versus bias of $\hat{\mu}_t$ in the 3072
    high-dimensional simulation settings ($p = 500$). When $s_t = 0$ (randomized experiment), the
    false coverage rate is well controlled. When $s_t > 0$
    (observational data), the FCR can be much larger
    than the nominal 5\% level, though FCR is is generally smaller
  when the nuisance functions are estimated by the lasso instead of
  random forest.}
  \label{fig:fcr-bias-high}
\end{figure}

\subsubsection{Coverage of WATE}
\label{sec:coverage-ate}

A somewhat surprising observation from the previous figures is that
although our selective inference procedure cannot guarantee selective
error control due to estimation error of $\mu_t$ and $\mu_y$, the
false coverage rate was not too much larger than the nominal level in
many simulation settings. In \Cref{fig:fcr-bias-high} we compare FCR
with coverage error of a weighted average treatment effect (WATE) in
all the simulation settings. \rev{The WATE is given by solving the
  optimization problem \eqref{eq:beta-star-causal} without effect
  modification, that is,
  \begin{equation} \label{eq:wate}
    \text{WATE} = \argmin_{\alpha} \sum_{i=1}^n \big[T_i -
    \mu_t(\mathbf{X}_i)\big]^2\big[\Delta(\mathbf{X}_i) -
    \alpha\big]^2 = \frac{\sum_{i=1}^n \big[T_i -
    \mu_t(\mathbf{X}_i)\big]^2\Delta(\mathbf{X}_i)}{\sum_{i=1}^n
    \big[T_i - \mu_t(\mathbf{X}_i)\big]^2}.
\end{equation}
Notice that as $n \to \infty$, we have
\[
  \text{WATE} \overset{p}{\to} \frac{\mathbb{E}[\mu_t(\bm X_i) (1 -
    \mu_t(\bm X_i)) \Delta(\bm X_i)]}{\mathbb{E}[\mu_t(\bm X_i) (1 -
    \mu_t(\bm X_i))]}
\]
by the law of large numbers. The population limit on the right hand side
is known as the optimally weighted average treatment effect
\citep{crump2006moving}; see also \citet{li2018balancing,zhao2019covariate}.
}

Confidence intervals of the WATE were obtained by a simple linear regression
of $Y_i - \hat{\mu}_y(\bm X_i)$ on $T_i - \hat{\mu}_t(\bm X_i)$. In
randomized settings ($s_t = 0$), both error rates were controlled at
the nominal level as expected. In observational settings ($s_t = 5$ or
$25$), coverage of WATE was very poor, sometimes completely missing the
target (coverage error is almost 100\%). In contrast, although there
is also no guarantee of controlling the FCR as we have shown previously,
the FCR was always smaller than the coverage error of WATE. This
suggests that the selective inference of effect modification may be
more robust to estimation error in the nuisance functions than the
semiparametric inference of WATE.

\begin{figure}[t]
  \centering
  \includegraphics[width = 0.9\textwidth]{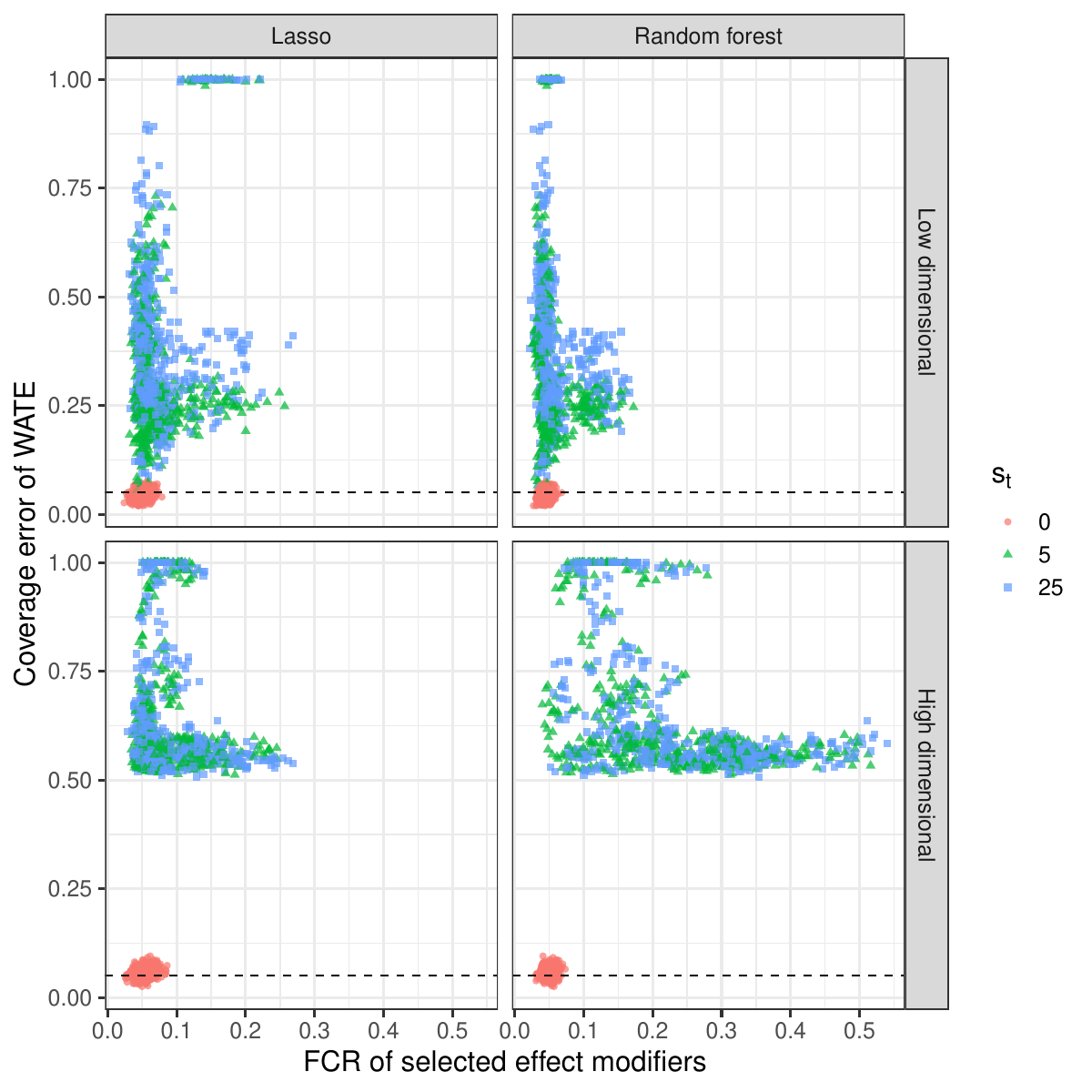}
  \caption{Coverage error of WATE versus false coverage rate (FCR) of
    the effect modifiers, stratified by the dimension of the data and
    the method used to estimate the nuisance functions. When $s_t = 0$
    (randomized experiment, red
    circles in the figure), both error rates are controlled at the
    nominal level (dashed line). When $s_t > 0$ (observational study),
    coverage of WATE is very poor, sometimes completely missing the
    target. There is also no guarantee of controlling FCR in general,
    though the FCR is smaller than the coverage error of WATE in
  the all the simulation settings.}
  \label{fig:fcr-bias-high}
\end{figure}

\subsection{Necessity and sufficiency of the rate assumptions}
\label{sec:valid-select-conf}

One of the main theoretical conclusions of this paper is that, when
the design and the outcome are observed with error, the selective
pivotal statistic is
still asymptotically valid as long as the classical semiparametric
rate assumptions
\Cref{assump:accuracy-treatment,assump:accuracy-outcome} are
satisfied. In the next simulation, we verify the sufficiency and
necessity of the crucial condition $\|\hat{\mu}_t - \mu_t\|_2
\cdot \|\hat{\mu}_y - \mu_y\|_2 =   o_p(n^{-1/2})$ in an
idealized setting. In this
simulation, the true design and the true outcome were generated by
\[
\mathbf{X}_i \in \mathbb{R}^{30} \overset{i.i.d.}{\sim}
\mathrm{N}(\mathbf{0},\mathbf{I}_{30}),~\mathrm{Y}_i \overset{i.i.d.}{\sim}
\mathrm{N}(\mathbf{X}_i^T \bm{\beta}, 1),~i=1,\dotsc,n,
\]
where $\bm{\beta} = (1,1,1,0,\dotsc,0)^T \in \mathbb{R}^{30}$.
Next, the
design and the outcome were perturbed by
\begin{equation} \label{eq:perturbation}
\mathbf{X}_i \mapsto \mathbf{X}_i \cdot (1 + n^{-\gamma}
D_{1i}),~Y_i \mapsto Y_i + n^{-\gamma} D_{2i},
\end{equation}
where $D_{1i}$ and $D_{2i}$ are independent standard Gaussian random
variables. Since the nuisance parameters $\mu_t$ and $\mu_y$ are
always estimated with error in \Cref{sec:select-infer-effect}, the $(1
+ n^{-\gamma} D_{1i})$ and $n^{-\gamma}D_{2i}$ terms were used to
simulate the estimation error. We used five different values of
$\gamma$ in this simulation, $\gamma =0.15$, $0.2$, $0.25$, $0.3$, or
$0.35$. Then we pretended the perturbed design and outcome were the truth
and used the pivot \eqref{eq:lasso-pivot} to
obtain selective $90\%$-confidence intervals, after solving a lasso
regression with $\lambda =
2\mathrm{E}[\|\mathbf{X}\bm{\epsilon}\|_{\infty}]$ that is commonly
used in high dimensional regression \citep{negahban2012unified}. We
also compared the performance of selective inference with the naive
inference that ignores model selection.

In \Cref{fig:perturbation}, we report the average false coverage
proportion in $100$ realizations for each $\gamma$ and sample size
$n$. The naive inference failed to control the false coverage rate in
every setting. For selective inference, a phase transition phenomenon
occurred at $\gamma = 0.25$: when $\gamma < 0.25$, the false coverage rate
increases as the sample size increases; when $\gamma > 0.25$, the
false coverage rate converges to the nominal $10\%$ level as the
sample size increases. This observation is consistent with the rate assumption
$\|\hat{\mu}_t - \mu_t\|_2 \cdot \|\hat{\mu}_y -
\mu_y\|_2 = o_p(n^{-1/2})$ in \Cref{assump:accuracy-outcome}.

\begin{figure}[t]
  \centering
    \includegraphics[width = 0.9\textwidth]{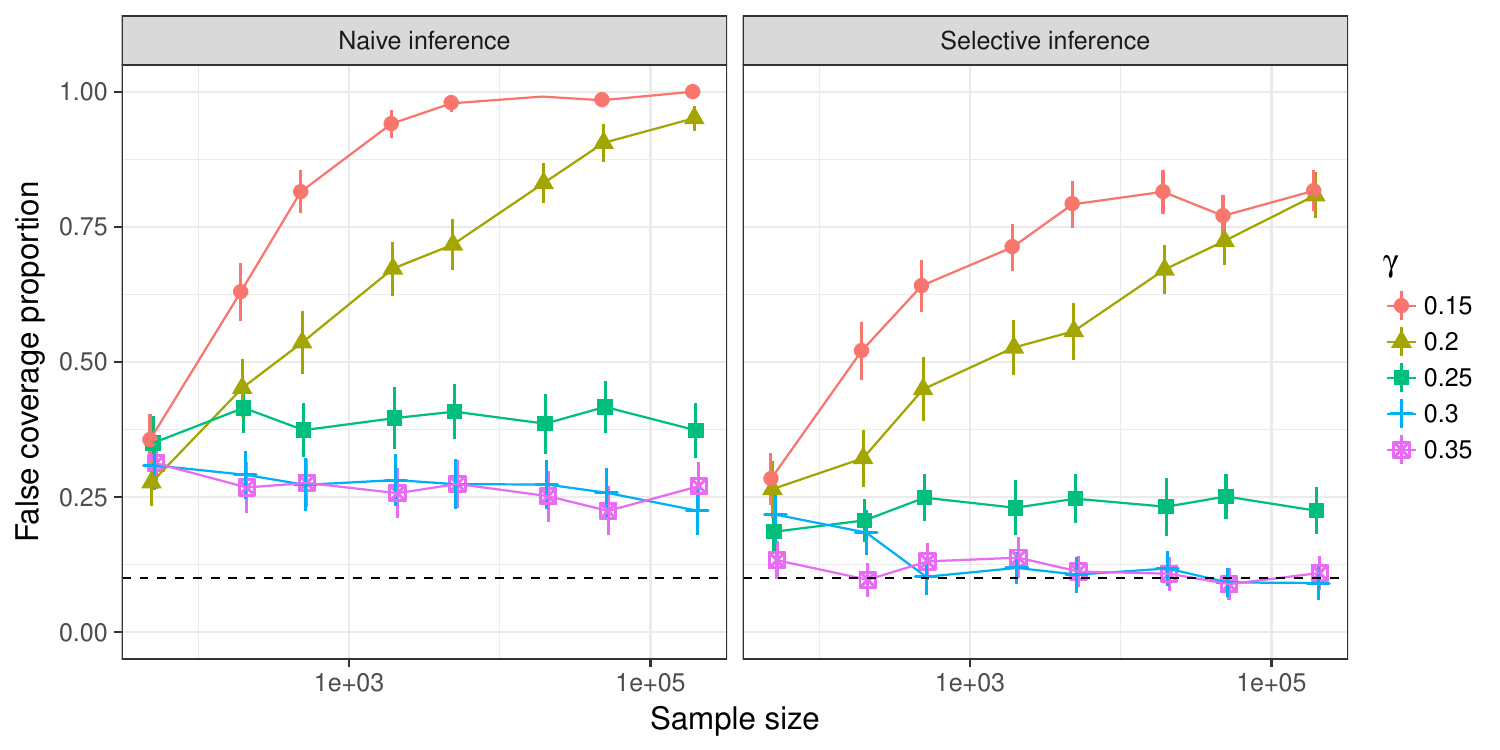}
  \caption{False coverage proportion under different strengths of
    perturbation and different sample sizes. Using naive inference
    that ignores the model is selected using the data, the false
    coverage rate is not controlled. Using selective
    inference, the false coverage proportion converges to the nominal
    $10\%$ level (the dashed horizontal line) if and only if $\gamma >
    0.25$.}
  \label{fig:perturbation}
\end{figure}


\section{Application: Overweight and systemic inflammation}
\label{sec:appl-obes-elev}

Finally we use an epidemiological study to demonstrate the method
proposed in this paper. \citet{visser1999elevated} studied the effect
of being overweight on low-grade systemic inflammation as measured by serum
C-reactive protein (CRP) level. Overweight was defined as body mass
index (BMI) greater than 25. Using the Third National Health and
Nutrition Examination Survey (NHANES III, 1988--1994), they found that
the CRP level is more likely to be elevated among overweight adults and the
effect is modified by gender and age group.

\subsection{Dataset and methods}
\label{sec:dataset-methods}

We obtain a more recent dataset from NHANES 2007--2008 and
2009--2010. We restricted to survey respondents who were not pregnant,
at least $21$ years old, and whose BMI and CRP levels are not
missing. Among the 10679 people left, 969 have missing income, 4 have
missing marital status, 15 have missing education, 1 has missing
information about frequent vigorous recreation, and 20 have no current
smoking information. To illustrate the method in this paper, we ignore
the entries with missing variables and end up with 9677
observations. The dataset and \texttt{R} code of our analysis are
included in the supplement.

\newcommand{\mytilde}{$\sim$}

The CRP level in the dataset ranges from $0.01$ mg/dL to $20.00$ mg/dL
and roughly follows a log-normal distribution (see the supplementary
file). Therefore we decided to
use $\log_2(\mathrm{CRP})$ as the response in the regression. We
use all the confounders identified in
\citet{visser1999elevated}, including gender, age, income, race,
marital status, education,
vigorous work activity (yes or no), vigorous recreation activities
(yes or no), ever smoked, number of cigarettes smoked in the last
month, estrogen usage, and if the survey respondent had bronchitis,
asthma, emphysema, thyroid problem, arthritis, heart attack, stroke, liver
condition, and gout. There are in total $20$ variables and some of
them are categorical. Using the \texttt{R} function
\texttt{model.matrix}, the design matrix $\mathbf{X}$ we use has
$9677$ rows and $30$ columns. To test the performance of selective inference
  in higher dimensions, we also consider a bigger model containing all
  the main effects and the first-order interactions of $\bm{X}$ (365 columns
  in total). We refer the reader to the supplementary file for more summary
  statistics of these variables.

We examine five different statistical analyses of effect modification
using this dataset:
\begin{enumerate}
\item Naive linear model: both $\eta(\mathbf{x})$ and
  $\Delta(\mathbf{x})$ are modeled by linear functions of
  $\mathbf{x}$.
\item Full model: in the following four models, the nuisance
  parameters ($\mu_y$ and $\mu_t$) are estimated by the random forest
  \citep{breiman2001random} (as implemented in the \texttt{R} package
  \texttt{randomForest}). In the full model, $\Delta(\mathbf{x})$ is
  modeled by the full linear model $\Delta(\mathbf{x}) \approx \alpha
  + \mathbf{x}^T \bm{\beta}$.
\item Univariate screening: $\Delta(\mathbf{x})$ is modeled by univariate
  linear model $\Delta(\mathbf{x}) \approx \alpha + x_j \beta_j$ for each
  $j=1,\dotsc,p$ (in the analysis we centered each column of
  $\mathbf{X}$, so the intercept $\alpha$ is the same).
\item Selected model: $\Delta(\mathbf{x}) \approx
  \alpha_{\hat{\mathcal{M}}} + \bm{x}_{\hat{\mathcal{M}}}^T
  \bm{\beta}_{\hat{\mathcal{M}}}$ where $\hat{\mathcal{M}}$ is selected by solving the
  lasso problem \eqref{eq:lasso} with $\lambda =
  1.1\times\mathrm{E}[\|\mathbf{X}\bm{\epsilon}\|_{\infty}]$ where
  $\bm{\epsilon} \sim \mathrm{N}(\mathbf{0},\hat{\sigma}^2\mathbf{I}_p)$ as
  suggested by \citet{negahban2012unified}. Then we used the pivotal
  statistic in \eqref{eq:asymptotic-validity} to make selective
  inference of $\bm{\beta}_{\hat{\mathcal{M}}}$. The noise variance
  $\sigma^2$ is estimated from a full model as suggested by
  \citet{lee2013exact}.
\item Data snooping model: this is the same as selected model except
  the statistical inference of $\bm{\beta}_{\hat{\mathcal{M}}}$
  ignores the fact that $\hat{\mathcal{M}}$ is selected using the data
  (known as data snooping).
\end{enumerate}

\subsection{Average treatment effect}
\label{sec:aver-treatm-effect}

We obtain estimates of the \rev{weighted average treatment effect (WATE) in
\eqref{eq:wate}} using the naive
linear model (method 1) and the full model (method 2). Since the last
four methods use the same estimated nuisance parameters, their
estimates of the WATE are identical. In other
words, their only difference is how effect modification is
modeled. The two estimates of WATE are quite close: using the naive
linear model (nuisance parameters are estimated by linear model), the
point estimate is 1.166 with 95\% confidence
interval [1.088, 1.244]; using the full model (parameters are
estimated by the random forest), the point estimate is 1.168 with
95\% confidence interval [1.071, 1.220].

\subsection{Effect modification}
\label{sec:effect-modification}

The results of the five methods are reported in
\Cref{tab:naive,tab:full,tab:univariate,tab:selective,tab:snooping},
respectively. The first three methods do not select a submodel, so the
coefficients of all $27$ regressors are reported. In the selective
inference (method 4), four regressors
(\texttt{Gender}, \texttt{Age}, \texttt{Stroke}, \texttt{Gout}) were
selected using the lasso when only main effects were used to
approximate $\Delta(\bm x)$. The last two regressors are replaced by \texttt{Age
$\times$ Vigorous work} and \texttt{Age $\times$ Stroke} in the lasso
when interaction terms are allowed. The corresponding partial
coefficients are reported for methods 4 and 5.

Next we discuss three observations that the reader may have already
noticed when comparing these results. First, although the naive linear
model (method 1) and
the full model (method 2) generate very similar estimates of the WATE, the
coefficients in their effect modification models are notably different
(compare \Cref{tab:naive} with \Cref{tab:full}). For example, the
estimated coefficient of \texttt{Gender} is $0.654$ using method 1 and $0.481$
using method 2. In general, the full model is more credible because
the nuisance parameters are more accurately estimated.

Second, the univariate screening (method 3) detects many
covariates that are very likely not the actual effect modifiers. Besides \texttt{Gender} and \texttt{Age}
themselves, all the other significant variables---\texttt{Marital (Widowed)},
\texttt{Marital (Never married)}, \texttt{Arthritis}, \texttt{Heart attack}, \texttt{Stroke}, and
\texttt{Gout}---are strongly correlated with \texttt{Gender} or
\texttt{Age} or both (the sample
correlations are at least $0.15$). When \texttt{Gender} and
\texttt{Age} are already in
the model (they were the first two active variables in the lasso solution path),
these variables are not subsequently selected by the lasso or are not
significant after adjusting for model selection.

Third, \texttt{Stroke} and \texttt{Gout} are selected by the lasso and they are not
significant using selective inference (method 4) in
\Cref{tab:selective-a}. However, they are significant
by using the naive inference that ignores model selection
(method 5) in \Cref{tab:snooping-a}. A similar phenomenon happens when
interaction terms are also used to model $\Delta(\bm x)$. \texttt{Age
$\times$ Vigorous work} and \texttt{Age $\times$ Stroke} are selected
by the lasso but were not significant in a selective inference
(\Cref{tab:selective-b}). The data snooping inference that ignores
model selection would find them to be significant
(\Cref{tab:snooping-b}). In general, non-selective inference
(data snooping) does not generate valid $p$-values and confidence
intervals. The example here demonstrates the practical importance of
selective inference, as stroke and gout would be reported as
significant otherwise.

\begin{table}[t]
  \centering
  \begin{tabular}{rrrrrl}
    \hline
    & Estimate & $p$-value & CI low & CI up &  \\
    \hline
    Gender (Female) & 0.654 & 0.000 & 0.488 & 0.821 & *** \\
    Age & -0.024 & 0.000 & -0.030 & -0.018 & *** \\
    Income & -0.019 & 0.495 & -0.072 & 0.035 &   \\
    Race (Hispanic) & 0.052 & 0.752 & -0.272 & 0.377 &   \\
    Race (White) & 0.166 & 0.196 & -0.086 & 0.418 &   \\
    Race (Black) & 0.376 & 0.010 & 0.089 & 0.664 & ** \\
    Race (Other) & 0.038 & 0.842 & -0.337 & 0.414 &   \\
    Marital (Widowed) & -0.083 & 0.593 & -0.389 & 0.223 &   \\
    Marital (Divorced) & 0.161 & 0.202 & -0.086 & 0.409 &   \\
    Marital (Separated) & -0.235 & 0.272 & -0.654 & 0.184 &   \\
    Marital (Never married) & 0.117 & 0.320 & -0.113 & 0.347 &   \\
    Marital (Living with partner) & -0.050 & 0.745 & -0.349 & 0.250 &   \\
    Education (9--11th grade) & 0.259 & 0.100 & -0.049 & 0.566 & . \\
    Education (High school) & 0.307 & 0.041 & 0.013 & 0.601 & * \\
    Education (Some college) & 0.296 & 0.052 & -0.002 & 0.593 & . \\
    Education (College grad.) & 0.316 & 0.054 & -0.006 & 0.638 & . \\
    Vigorous work & -0.019 & 0.854 & -0.216 & 0.179 &   \\
    Vigorous recreation & -0.323 & 0.001 & -0.521 & -0.125 & *** \\
    Ever smoked & -0.067 & 0.447 & -0.239 & 0.105 &   \\
    \# Cigarettes last month & -0.000 & 0.437 & -0.000 & 0.000 &   \\
    Estrogen & -0.645 & 0.002 & -1.063 & -0.228 & ** \\
    Bronchitis & -0.092 & 0.725 & -0.603 & 0.420 &   \\
    Asthma & 0.193 & 0.230 & -0.122 & 0.509 &   \\
    Emphysema & 0.045 & 0.862 & -0.464 & 0.554 &   \\
    Thyroid problem & 0.122 & 0.438 & -0.187 & 0.431 &   \\
    Arthritis & -0.046 & 0.644 & -0.240 & 0.148 &   \\
    Heart attack & -0.178 & 0.393 & -0.586 & 0.230 &   \\
    Stroke & -0.364 & 0.090 & -0.785 & 0.057 & . \\
    Liver condition & -0.332 & 0.311 & -0.973 & 0.310 &   \\
    Gout & -0.584 & 0.012 & -1.040 & -0.129 & * \\
    \hline
  \end{tabular}
  \caption{Results of the naive linear model (method 1) where $\eta(\bm{x})$ and
    $\Delta(\bm{x})$ are modeled by linear functions of $\bm{x}$. In
    other words, it is assumed that $Y_i = \gamma_0 + \mathbf{X}_i^T \bm{\gamma} +
    T_i (\alpha + \mathbf{X}_i^T \bm{\beta})$ and the reported coefficients are
    $\bm{\beta}$. The reported values are the point estimates,
    $p$-values and confidence intervals (CI) of each entry of
    $\bm{\beta}$. In general, the results are less credible because
    the linear model can be misspecified.}
  \label{tab:naive}
\end{table}

\begin{table}[t]
  \centering
  \begin{tabular}{rrrrrl}
    \hline
    & Estimate & $p$-value & CI low & CI up &  \\
    \hline
    Gender (Female) & 0.467 & 0.000 & 0.307 & 0.628 & *** \\
    Age & -0.021 & 0.000 & -0.027 & -0.016 & *** \\
    Income & 0.000 & 0.986 & -0.053 & 0.054 &   \\
    Race (Hispanic) & 0.209 & 0.183 & -0.099 & 0.517 &   \\
    Race (White) & 0.266 & 0.028 & 0.028 & 0.504 & * \\
    Race (Black) & 0.384 & 0.006 & 0.112 & 0.656 & ** \\
    Race (Other) & 0.136 & 0.482 & -0.244 & 0.516 &   \\
    Marital (Widowed) & -0.103 & 0.496 & -0.401 & 0.194 &   \\
    Marital (Divorced) & 0.152 & 0.219 & -0.090 & 0.394 &   \\
    Marital (Separated) & -0.157 & 0.501 & -0.613 & 0.299 &   \\
    Marital (Never married) & 0.116 & 0.316 & -0.111 & 0.344 &   \\
    Marital (Living with partner) & -0.168 & 0.271 & -0.468 & 0.131 &   \\
    Education (9--11th grade) & 0.036 & 0.812 & -0.262 & 0.334 &   \\
    Education (High School) & 0.073 & 0.616 & -0.213 & 0.359 &   \\
    Education (Some college) & 0.007 & 0.962 & -0.282 & 0.296 &   \\
    Education (College graduates) & 0.074 & 0.647 & -0.243 & 0.392 &   \\
    Vigorous work & 0.012 & 0.904 & -0.185 & 0.209 &   \\
    Vigorous recreation & -0.314 & 0.002 & -0.514 & -0.113 & ** \\
    Ever smoked & 0.005 & 0.957 & -0.162 & 0.171 &   \\
    \# Cigarettes last month & -0.000 & 0.565 & -0.000 & 0.000 &   \\
    Estrogen & -0.608 & 0.006 & -1.045 & -0.171 & ** \\
    Bronchitis & -0.144 & 0.609 & -0.695 & 0.407 &   \\
    Asthma & 0.197 & 0.242 & -0.133 & 0.527 &   \\
    Emphysema & -0.006 & 0.984 & -0.541 & 0.530 &   \\
    Thyroid & 0.134 & 0.380 & -0.165 & 0.433 &   \\
    Arthritis & -0.076 & 0.427 & -0.264 & 0.112 &   \\
    Heart attack & -0.214 & 0.321 & -0.636 & 0.208 &   \\
    Stroke & -0.492 & 0.023 & -0.917 & -0.067 & * \\
    Liver condition & -0.190 & 0.603 & -0.908 & 0.527 &   \\
    Gout & -0.489 & 0.034 & -0.940 & -0.037 & * \\
    \hline
  \end{tabular}
  \caption{Results of the full model (method 2) where $\mu_y(\bm{x})$ and
    $\mu_t(\bm{x})$ are estimated by the random forest and
    $\Delta(\bm{x})$ are modeled by $\Delta(\bm{x}) \approx \alpha +
    \bm{x}^T \bm{\beta}$. The reported coefficients are
    $\bm{\beta}$, which is the best linear approximation of
    $\Delta(\bm{x})$ in the sense of \eqref{eq:beta-star-causal}. Since
    there are $27$ regressors in total and many of them are strongly
    correlated, it is difficult to interpret these coefficients.}
  \label{tab:full}
\end{table}

\begin{table}[t]
  \centering
  \begin{tabular}{rrrrrl}
    \hline
    & Estimate & $p$-value & CI low & CI up &  \\
    \hline
    Gender (Female) & 0.494 & 0.000 & 0.350 & 0.639 & *** \\
    Age & -0.021 & 0.000 & -0.025 & -0.017 & *** \\
    Income & -0.013 & 0.582 & -0.057 & 0.032 &   \\
    Race (Hispanic) & 0.108 & 0.399 & -0.143 & 0.359 &   \\
    Race (White) & -0.060 & 0.416 & -0.204 & 0.084 &   \\
    Race (Black) & 0.175 & 0.076 & -0.018 & 0.368 & . \\
    Race (Other) & -0.069 & 0.683 & -0.399 & 0.261 &   \\
    Marital (Widowed) & -0.529 & 0.000 & -0.790 & -0.268 & *** \\
    Marital (Divorced) & 0.088 & 0.455 & -0.143 & 0.320 &   \\
    Marital (Separated) & -0.135 & 0.554 & -0.583 & 0.313 &   \\
    Marital (Never married) & 0.454 & 0.000 & 0.257 & 0.651 & *** \\
    Marital (Living with partner) & 0.073 & 0.609 & -0.206 & 0.351 &   \\
    Education (9--11th grade) & -0.005 & 0.964 & -0.200 & 0.191 &   \\
    Education (High school) & 0.068 & 0.433 & -0.102 & 0.238 &   \\
    Education (Some college) & 0.090 & 0.278 & -0.073 & 0.254 &   \\
    Education (College graduates) & 0.005 & 0.953 & -0.165 & 0.175 &   \\
    Vigorous work & 0.043 & 0.654 & -0.145 & 0.230 &   \\
    Vigorous recreation & -0.048 & 0.602 & -0.229 & 0.133 &   \\
    Ever smokes & -0.138 & 0.062 & -0.282 & 0.007 & . \\
    \# Cigarettes last month & 0.000 & 0.651 & -0.000 & 0.000 &   \\
    Estrogen & 0.070 & 0.746 & -0.352 & 0.491 &   \\
    Bronchitis & -0.194 & 0.457 & -0.704 & 0.317 &   \\
    Asthma & 0.229 & 0.155 & -0.086 & 0.544 &   \\
    Emphysema & -0.453 & 0.067 & -0.937 & 0.032 & . \\
    Thyroid & -0.007 & 0.962 & -0.296 & 0.282 &   \\
    Arthritis & -0.410 & 0.000 & -0.576 & -0.244 & *** \\
    Heart attack & -0.730 & 0.000 & -1.137 & -0.324 & *** \\
    Stroke & -0.919 & 0.000 & -1.332 & -0.506 & *** \\
    Liver condition & -0.526 & 0.149 & -1.240 & 0.189 &   \\
    Gout & -0.987 & 0.000 & -1.428 & -0.546 & *** \\
    \hline
  \end{tabular}
  \caption{Results of the univariate screening (method 3) where $\mu_y(\bm{x})$ and
    $\mu_t(\bm{x})$ are estimated by the random forest and
    $\Delta(\bm{x})$ are then modeled by $\Delta(\bm{x}) \approx \alpha +
    x_j^T \beta_j$ for each $j = 1,\dotsc,p$. This simple method can be used
    to detect potential effect modifiers. However, all the other significant
    regressors are strongly correlated with gender or Age, so it is very
    likely that they are not the actual effect modifiers.}
  \label{tab:univariate}
\end{table}
\begin{table}[t]
  \centering
  \begin{subtable}[t]{\textwidth}
    \centering
    \begin{tabular}{rrrrrl}
      \hline
      & Estimate & $p$-value & CI low & CI up &  \\
      \hline
      Gender (Female) & 0.476 & 0.000 & 0.330 & 0.624 & *** \\
      Age & -0.019 & 0.000 & -0.024 & -0.015 & *** \\
      Stroke & -0.515 & 0.311 & -0.899 & 1.256 &   \\
      Gout & -0.475 & 0.493 & -0.852 & 2.295 &   \\
      \hline
    \end{tabular}
    \caption{Using only main effects to model effect modification.}
    \label{tab:selective-a}
  \end{subtable}
  \begin{subtable}[t]{\textwidth}
    \centering
    \begin{tabular}{rrrrrl}
      \hline
      & Estimate & $p$-value & CI low & CI up &  \\
      \hline
      Gender (Female) & 0.471 & 0.000 & 0.323 & 0.618 & *** \\
      Age & -0.020 & 0.000 & -0.024 & -0.016 & *** \\
      Age $\times$ Vigorous recreation & 0.018 & 0.371 & -0.052 & 0.027 &   \\
      Age $\times$ Stroke & -0.036 & 0.069 & -0.054 & 0.014 & . \\
      \hline
    \end{tabular}
    \caption{Using main effects and first-order interactions to model
      effect modification.}
    \label{tab:selective-b}
  \end{subtable}
  \caption{Results of the selective inference (method 4) where $\mu_y(\bm{x})$ and
    $\mu_t(\bm{x})$ are estimated by the random forest and
    $\Delta(\bm{x})$ are modeled by $\Delta(\bm{x}) \approx \alpha +
    \mathbf{x}_{\hat{\mathcal{M}}}^T \bm{\beta}_{\hat{\mathcal{M}}}$
    where $\hat{\hat{\mathcal{M}}}$ is selected after fitting a lasso
    with the main effects (\Cref{tab:selective-a}) or the
    main effects and first-order interactions (\Cref{tab:selective-b}). The
    selective $p$-values and confidence intervals are obtained using
    the pivotal statistic \eqref{eq:asymptotic-validity}, which are
    asymptotically valid given the assumptions in this paper.}
  \label{tab:selective}
\end{table}
\begin{table}[t]
  \centering
  \begin{subtable}[t]{\textwidth}
    \centering
    \begin{tabular}{rrrrrl}
      \hline
      & Estimate & $p$-value & CI low & CI up &  \\
      \hline
      Gender (Female) & 0.476 & 0.000 & 0.332 & 0.620 & *** \\
      Age & -0.019 & 0.000 & -0.023 & -0.015 & *** \\
      Stroke & -0.514 & 0.016 & -0.933 & -0.096 & * \\
      Gout & -0.473 & 0.038 & -0.919 & -0.026 & * \\
      \hline
    \end{tabular}
    \caption{Using only main effects to model effect modification.}
    \label{tab:snooping-a}
  \end{subtable}
  \begin{subtable}[t]{\textwidth}
    \centering
    \begin{tabular}{rrrrrl}
      \hline
      & Estimate & $p$-value & CI low & CI up &  \\
      \hline
      Gender (Female) & 0.471 & 0.000 & 0.327 & 0.615 & *** \\
      Age & -0.020 & 0.000 & -0.024 & -0.016 & *** \\
      Age $\times$ Vigorous recreation & 0.018 & 0.001 & 0.008 & 0.028 & *** \\
      Age $\times$ Stroke & -0.036 & 0.000 & -0.055 & -0.017 & *** \\
      \hline
    \end{tabular}
    \caption{Using main effects and first-order interactions to model
      effect modification.}
    \label{tab:snooping-b}
  \end{subtable}
  \caption{Results of data snooping (method 5). Everything is the same
    as the selective inference (method 4) except that
    $\hat{\mathcal{M}}$ is treated as given. The $p$-values and
    confidence intervals are not valid because the bias
    due to model selection was not taken into account.}
  \label{tab:snooping}
\end{table}

\section{Discussion}
\label{sec:discussion}

\subsection{When is selective inference a good approach for effect modification?}
\label{sec:when-should-select}

In \Cref{sec:introduction} we have compared accuracy and
interpretability of different approaches to modeling effect
modification. The machine learning approaches usually approximate the
conditional average treatment effect $\Delta(\mathbf{x})$ better but
are difficult to interpret. The univariate regressions find
significantly covariates correlated with $\Delta(\mathbf{x})$, but
that correlation can vanish after conditioning on other covariates as
illustrated in the example in \Cref{sec:appl-obes-elev}.
The selective inference approach provides an appealing tradeoff
between accuracy and interpretability and is a good approach for
modeling effect modification when interpretability is important in the
study, for example, when the goal is to generate new scientific
hypotheses or to assist clinicians to make intelligent treatment decisions.

The simulations in \Cref{sec:simulation} show that the selective
inference for effect modification can be trusted in randomized
experiments, but does not guarantee selective error control in
observational studies, especially when the dimension of confounders is
high. The consoling news is that the selective error is often not too
much higher than the nominal level, and is usually smaller than the
coverage error of the weighted ATE. Thus with observational
and high dimensional data, the results of the selective inference of
effect modification should be interpreted with caution, although they
are usually more reliable than the inference of the weighted ATE.

When the dimension of $\bm X$ is high, an alternative approach to
high-dimensional regression inference is the debiased lasso
\citep{zhang2014confidence,van2014asymptotically,javanmard2014confidence}. Compared
to the debiased lasso, the selective inference approach does not require
sparsity of the true regression coefficients. Instead, it shifts the
inferential target to the best linear approximation using the selected
covariates \citep{berk2013valid}. It is also possible to make
selective inference for the regression coefficient in the full model
\citep{liu2018more}. Notice that lasso is used for different
purposes in these two approaches: the debiased lasso uses the
lasso to estimate the regression coefficients, while selective
inference only uses the lasso to select a submodel. In principle, we think
our method is most useful in randomized experiments or observational
studies with a moderate number (e.g.\ $p = 50$) of covariates. Our
method can be used in observational and high-dimensional setting,
though in the $p \gg n$ case, it seems inevitable that strong
sparsity assumptions about the true causal model are needed (in which case the debiased
lasso would be more appealing). Another interesting
approach to high-dimensional regression is the knockoff filter
that aims at controlling the false discovery rate
\citep{barber2015controlling}.

There are two situations where selective inference should not be
used to model effect modification. The first is when prediction
accuracy is the only goal. This can happen if we want to learn the optimal
treatment regime and are not concerned about interpretability at
all. In this case, machine learning methods such as outcome-weighted
learning \citep{zhao2012estimating} should be
used, though this black box approach also raises many policy
concerns and challenges \citep{price2014black}. Although selective
inference can be used to generate interpretable decision rules, a
generally more efficient low-complexity policy can be learned by
directly maximizing the utility \citep{athey2017efficient}.

The second situation is when we are interested in the causal effect of
both the treatment and the discovered effect modifiers. Since we do
not control for
confounding between the effect modifiers and the outcome, selective
inference nor any other method that does not control for such
confounding can be used to estimate the causal
effect of the effect modifiers. In other words, effect modifiers
explain variation of the causal effect and may themselves be non-causal.
Nonetheless, the advocated selective
inference approach is useful for post-hoc discovery of important
effect modifiers which are usually informative proxies to the
underlying causes \citep{vanderweele2007four}. See \citet[Section
9.6]{vanderweele2015explanation} and the references therein for more
discussion.

\subsection{Assumptions in the paper}
\label{sec:assumptions-paper}

Our main theoretical result (\Cref{thm:asymptotic-validity}) hinges
on a number of assumptions. Here we explain when they are reasonable
and discuss their implications in more detail.

\Cref{assump:basic} is fundamental to causal inference. It transforms
the estimation of causal effect into a regression
problem. Unconfoundedness
(\Cref{assump:basic}A) is crucial to identify the causal effect for
observational studies, but it is unverifiable using observational
data. It would be an interesting further investigation to study the
sensitivity of the proposed method to the unconfoundedness assumption.

\Cref{assump:accuracy-treatment,assump:accuracy-outcome} are rate
assumptions for the estimated nuisance parameters. The product structure in
\Cref{assump:accuracy-outcome} is closely related to doubly robust
estimation \citep[see
e.g.][]{bang2005doubly,chernozhukov2016double}. \Cref{assump:accuracy-treatment,assump:accuracy-outcome} are essential
to the semiparametric problem under consideration and are satisfied by
using, for example, kernel smoothing with optimal bandwidth when $p \le
3$. However, there is little interest for selective inference in such a
low-dimensional problem. In general, no method can guarantee the rate
assumptions are universally satisfied, an issue present in all
observational studies. This is why we have recommended to use machine
learning methods such as the random forest to estimate the nuisance
parameters, as they usually
have much better prediction accuracy than conventional parametric
models. This practical advice is inspired by
\citet{vanderlaan2011} and \citet{chernozhukov2016double}. Recently
\citet{van2017generally} proposed a nonparametric regression estimator
called the highly adaptive lasso (HAL) which converges faster than
$n^{-1/4}$ for functions with bounded variation norm.

\Cref{assump:model-size} strongly restricts the size of the effect modification
model. We believe it is indispensable in our approach and other
semiparametric regressions to control the complexity of the parametric
part. \Cref{assump:model-size} is also
used by \citet{tian2017asymptotics} to relax the Gaussianity assumption of the
noise. \Cref{assump:design} assumes the selected design matrix is not
collinear and is a sparse eigenvalue assumption in
high-dimensional regression \citep{buhlmann2011statistics}. It is needed
to define the partial regression coefficient $\hat{\bm
  \beta}_{\mathcal{M}}^{*}$ in \eqref{eq:beta-star-causal}. The boundedness assumptions in
\Cref{assump:support-X,assump:smooth-pivot,assump:kkt} are technical
assumptions for the asymptotic analysis. Similar assumptions can be
found in \citet{tian2017asymptotics} that are used to prove the
asymptotics under non-Gaussian error. In our experience, the inversion
of the pivot (to obtain selective confidence interval) is often
unstable when \Cref{assump:smooth-pivot} is not satisfied. In this
case, \citet{tian2015selective} proposed to smooth out the selection
event by injecting noise to the outcome. \Cref{assump:kkt} is used to
ensure that the selection event using the
estimated $\mu_y$ and $\mu_t$ is with high probability the same as the
selection event using the true $\mu_y$ and $\mu_t$. We expect that
\Cref{assump:smooth-pivot,assump:kkt} can be greatly weakened using
the randomized response approach to selective inference of
\citet{tian2015selective}.

In our main Theorem we also assumed the noise is homoskedastic and
Gaussian. This simplifies the proof as we can directly use the exact
selective inference \Cref{lem:lasso-selective-inference} derived by
\citet{lee2013exact}. We expect that this assumption can
be weakened \citep[see e.g.][]{tian2015selective} as we only need asymptotic validity
of the pivot when $\mu_y(\mathbf{x})$ is known (see the proof of
\Cref{thm:asymptotic-validity} in the \Cref{sec:proof-crefthm:-valid}).

In our theoretical investigation, cross-fitting \citep{schick1986asymptotically} is used
to bound the remainder terms like $\sum_{i=1}^n (T_i -
\mu_{t}(\bm X_i))(\hat{\mu}_{y}(\bm X_i) - \mu_y(\bm X_i))$. However, we find in our
simulations that using this theory-oriented technique
actually deterioriates the coverage rate of the weighted ATE and the effect
modification. \revv{In our simulations,} cross-fitting behvaes poorly
especially when random forests are used to estimate the nuisance
functions. \revv{This is mainly due to systematic bias in the tail observations, a
common problem for
tree-based machine learning methods even when the out-of-bag
  sample is used for prediction \citep{friedberg2018local}. The
bias is exacerbated by the reduction of sample size with
the two-fold cross-fitting we implemented, which might be why
cross-fitting does not work well with random forests in our simulations.}
In contrast, when the lasso is used to estimate the
nuisance functions in our simulations, using cross-fitting does not
change the overall performance substantially \revv{in our
  simulations}.


\subsection{Future directions}
\label{sec:future-directions}

We have focused on semiparametric regression with additive noise in this paper so
\citet{robinson1988root}'s transformation can be applied. In general,
many interesting estimands in causal inference \rev{and other statistical
problems} can be defined by estimating equations. \rev{While the
present paper was being drafted and reviewed, a parallel literature
has used the idea of Neyman orthogonalization (a generalization of
Robinson's transformation) in a number of related problems, including
inference about low-dimensional projections of causal functions
\citep{semenova2017estimation2,chernozhukov2018generic}, debiased-lasso
inference for coefficients of the CATE $\Delta(\bm X)$ assuming a sparse and
linear model for $\Delta(\bm X)$ \citep{semenova2017estimation},
nonparametric estimation of the CATE $\Delta(\bm X)$
\citep{nie2017quasi}, regularized estimation of high-dimensional
semiparametric nonlinear models \citep{chernozhukov2018plug}, and
statistical learning with nuisance parameters
\citep{foster2019orthogonal}. Our work is the first to consider
post-selection inference for this type of problems. For future
research it would be very interesting to develop selective inference in
more general semiparametric problems, beyond the effect modification
problem considered here.} Other possible future
directions include selective inference for non-linear models (e.g.\
trees) of $\Delta(\bm x)$ and extending the method in this paper to
longitudinal problems.



\bibliographystyle{abbrvnat}
\bibliography{ref}

\appendix

\section{Proofs}
\label{sec:proof}

\subsection{Proof of \Cref{thm:hat-star}}
\label{sec:proof-crefthm:h-star}

We first prove a Lemma that shows $\bm{\beta}^{*}_{\mathcal{M}}$ is
bounded.

\begin{lemma} \label{lem:bounded-beta-star}
    Under \Cref{assump:basic,assump:support-X} and $\mathrm{E}[\bm X_{i,\mathcal{M}} \bm
  X_{i,\mathcal{M}}^T] \succeq (1/C) \bm{I}_{|\mathcal{M}|}$,
  $\|{\bm{\beta}}^{*}_{\mathcal{M}}\|_{\infty} = O_p(1)$.
\end{lemma}
\begin{proof}
See \Cref{sec:proof-crefl-tilde} below.
\end{proof}

Next we prove \Cref{thm:hat-star}.
For simplicity we suppress the subscript $\mathcal{M}$ since it is
a fixed set in \Cref{sec:infer-fixed-model}. Let
\[
\bm{\psi}(\bm{\beta},\mu_t) = \frac{1}{n} \sum_{i=1}^n (T_i - \mu_t(\bm{X}_i))^2
(\Delta(\mathbf{X}_i) - \mathbf{X}_i^T \bm{\beta}) \bm{X}_i.
\]
The first-order conditions for $\bm{\beta}^{*}$ and
$\tilde{\bm{\beta}}$ are $\bm{\psi}(\bm{\beta}^{*},\mu_t) = \bm{0}$ and
$\bm{\psi}(\tilde{\bm{\beta}},\hat{\mu}_t) = \bm{0}$. Notice that
$\psi$ is a linear function of $\bm{\beta}$, so
\[
\begin{split}
  \bm{0} &= \sqrt{n} \bm{\psi}(\tilde{\bm{\beta}},\hat{\mu}_t) \\
  &= \sqrt{n} \bm{\psi}(\bm{\beta}^{*},\hat{\mu}_t) + \Big[\frac{1}{n}
  \sum_{i=1}^n (T_i - \hat{\mu}_t(\bm{X}_i))^2 \bm{X}_i\bm{X}_i^T\Big]
  \sqrt{n} (\tilde{\bm{\beta}} - \bm{\beta}^{*}).
\end{split}
\]
Similar to the proof of \Cref{lem:bounded-beta-star}, the term in the squared bracket converges to
$\mathrm{E}[\mathrm{Var}(T_i|\mathbf{X}_i) \cdot \mathbf{X}_i \mathbf{X}_i^T]$
which is positive definite by assumption. Thus it suffices to show
$\sqrt{n} \bm{\psi}(\bm{\beta}^{*},\hat{\mu}_t) \overset{p}{\to}
0$. This is true because
\begin{equation} \label{eq:psi-beta}
  \begin{split}
    &\sqrt{n} \bm{\psi}(\bm{\beta}^{*},\hat{\mu}_t) \\
    =& \frac{1}{\sqrt{n}} \sum_{i=1}^n \big[(T_i - \mu_t(\mathbf{X}_i) +
    \mu_t(\mathbf{X}_i) - \hat{\mu}_t(\mathbf{X}_i)\big]^2
    (\Delta(\mathbf{X}_i) - \mathbf{X}_i^T \bm{\beta}^{*}) \mathbf{X}_i
    \\
    =& \frac{1}{\sqrt{n}} \sum_{i=1}^n (T_i - \mu_t(\mathbf{X}_i))^2 \cdot
    (\Delta(\mathbf{X}_i) - \mathbf{X}_i^T \bm{\beta}^{*})
    \mathbf{X}_i \\
    &+ \frac{1}{\sqrt{n}} \sum_{i=1}^n 2
    (\mu_t(\mathbf{X}_i) - \hat{\mu}_t(\mathbf{X}_i)) (T_i -
    \mu_t(\mathbf{X}_i)) \cdot (\Delta(\mathbf{X}_i) - \mathbf{X}_i^T
    \bm{\beta}^{*}) \mathbf{X}_i \\
    &+ \frac{1}{\sqrt{n}} \sum_{i=1}^n (\mu_t(\mathbf{X}_i) -
    \hat{\mu}_t(\mathbf{X}_i))^2 \cdot (\Delta(\mathbf{X}_i) - \mathbf{X}_i^T
    \bm{\beta}^{*}) \mathbf{X}_i.
  \end{split}
\end{equation}
The first term is $\bm{0}$ because $\bm{\psi}(\bm{\beta}^{*},\mu_t) =
\bm{0}$. The second term is $\bm{o}_p(n^{-1/4})$ because $\|\mu_t -
\hat{\mu}_t\|_2 = o_p(n^{-1/4})$ and the rest is an i.i.d.\ sum with mean
$\mathrm{E}[(T_i-\mu_t(\bm{X}_i)) (\Delta(\mathbf{X}_i) -
\mathbf{X}_i^T \bm{\beta}^{*}) \mathbf{X}_i] = \bm{0}$. The third term
is $\bm{o}_p(1)$ because of the rate asumption $\|\mu_t - \hat{\mu}_t\|_2 =
o_p(n^{-1/4})$ and boundedness of $\Delta(\bm X)$, $\bm X$, and $\bm
\beta^{*}$.

\subsection{Proof of \Cref{thm:beta-fixed}}
\label{sec:proof-crefthm:b-tild}

\begin{lemma} \label{thm:beta-tilde-M}
  Under
  \Cref{assump:basic,assump:accuracy-treatment,assump:support-X,assump:accuracy-outcome},
  we have
  \[
  \Big(\sum_{i=1}^n (T_i - \hat{\mu}_t(\mathbf{X}_{i}))^2
  \mathbf{X}_{i,\mathcal{M}} \mathbf{X}_{i, \mathcal{M}}^T\Big)^{-1/2} (\hat{\bm{\beta}}_{\mathcal{M}} -
  \tilde{\bm{\beta}}_{\mathcal{M}}) \overset{d}{\to} \mathrm{N}(0,\sigma^2 \mathbf{I}_{|\mathcal{M}|}).
\]
\end{lemma}

Combining \Cref{thm:hat-star} and \Cref{thm:beta-tilde-M}, we obtain
the asymptotic inference of $\bm{\beta}_{\mathcal{M}}^{*}$ in
\Cref{thm:beta-fixed}. Next we prove \Cref{thm:beta-tilde-M}.

Like in \Cref{sec:proof-crefthm:h-star} we suppress the subscript
$\mathcal{M}$ for simplicity of notation. Denote $\mu_{yi} =
\mu_y(\mathbf{X}_i)$, $\mu_{ti} = \mu_t(\mathbf{X}_i)$ and similarly
for $\hat{\mu}_{yi}$ and $\hat{\mu}_{ti}$. Since $\hat{\bm{\beta}}$ is
the least squares solution, we have
\[
\begin{split}
  \hat{\bm{\beta}} &= \Big[\frac{1}{n}
  \sum_{i=1}^n (T_i - \hat{\mu}_{ti})^2
  \bm{X}_i\bm{X}_i^T\Big]^{-1} \Big[ \frac{1}{n} \sum_{i=1}^n (T_i -
  \hat{\mu}_{ti}) \mathbf{X}_i (y_i -
  \hat{\mu}_{yi}) \Big] \\
  &= \Big[\frac{1}{n}
  \sum_{i=1}^n (T_i - \hat{\mu}_{ti})^2
  \bm{X}_i\bm{X}_i^T\Big]^{-1} \Big\{ \frac{1}{n} \sum_{i=1}^n (T_i -
  \hat{\mu}_{ti}) \mathbf{X}_i \big[\mu_{yi} -
  \hat{\mu}_{yi} + (T_i - \mu_{ti}) \Delta(\mathbf{X}_i) +
  \epsilon_i\big] \Big\} \\
  &= \tilde{\bm{\beta}} + \Big[\frac{1}{n}
  \sum_{i=1}^n (T_i - \hat{\mu}_{ti})^2
  \bm{X}_i\bm{X}_i^T\Big]^{-1}\Big[ \frac{1}{n} \sum_{i=1}^n (T_i -
  \hat{\mu}_{ti}) \epsilon_i \mathbf{X}_i \Big]  \\
  &\qquad +\Big[\frac{1}{n}
  \sum_{i=1}^n (T_i - \hat{\mu}_{ti})^2
  \bm{X}_i\bm{X}_i^T\Big]^{-1} \Big\{\frac{1}{n} \sum_{i=1}^n (T_i
  - \mu_{ti} + \mu_{ti} -
  \hat{\mu}_{ti}) \mathbf{X}_i \big[\mu_{yi} -
  \hat{\mu}_{yi} + (\hat{\mu}_{ti} - \mu_{ti})
  \Delta(\mathbf{X}_i)\big] \Big\} \\
  &= \tilde{\bm{\beta}} + \Big[\frac{1}{n}
  \sum_{i=1}^n (T_i - \hat{\mu}_{ti})^2
  \bm{X}_i\bm{X}_i^T\Big]^{-1}\Big[ \frac{1}{n} \sum_{i=1}^n (T_i -
  \hat{\mu}_{ti}) \epsilon_i \mathbf{X}_i \Big] + o_p(n^{-1/2}). \\
\end{split}
\]
In the last equation, the residual terms are smaller than $n^{-1/2}$
because of the rate assumptions in the Lemma (each term can be
analyzed as in the proof of \Cref{thm:hat-star}).

\subsection{Proof of \Cref{lem:beta-tilde-star-random}}
\label{sec:proof-crefl-tilde}

We first prove a Lemma. Denote $\kappa(\mathbf{Z})$ be the
all the eigen-values of a square matrix $\mathbf{Z}$.

\begin{lemma} \label{lem:eta}
  Under \Cref{assump:basic,assump:support-X,assump:model-size,assump:design}, with probability going
  to $1$, for any $k$,
  \[
  1/(2 C^2) \le \kappa\Big((1/n)\tilde{\bm{X}}^T_{\cdot
    \hat{\mathcal{M}}}\tilde{\bm{X}}_{\cdot
    \hat{\mathcal{M}}}\Big) \le 2 m C^3.
  \]
  Therefore $\tilde{\bm{\eta}}^T \tilde{\bm{\eta}} = \Theta_p(1/n)$,
  meaning that for any $\epsilon > 0$, there exists a constant $C > 1$ such that
  $\mathrm{P}(1/(Cn) \le \tilde{\bm{\eta}}^T \tilde{\bm{\eta}} \le
  C/n) \ge 1- \epsilon$ for sufficiently large $n$.
\end{lemma}

\begin{proof}
  For the first result, by
  \Cref{assump:model-size}, we only need to bound, for every
  $|\mathcal{M}| \le m$, the eigenvalues of $((1/n)\tilde{\bm{X}}^T_{\cdot
    {\mathcal{M}}}\tilde{\bm{X}}_{\cdot {\mathcal{M}}})$. This matrix converges
  to $\mathrm{E}[\mathrm{Var}(T_i|\bm{X}_{i}) \cdot
  \bm{X}_{i,\mathcal{M}} \bm{X}_{i,\mathcal{M}}^T]$, whose eigenvalues
  are bounded by
  \[
  \begin{split}
    \kappa \Big( \mathrm{E}[\mathrm{Var}(T_i|\bm{X}_{i}) \cdot
    \bm{X}_{i,\mathcal{M}} \bm{X}_{i,\mathcal{M}}^T] \Big) &\in
    \bigg[(1/C) \cdot \kappa_{\mathrm{min}} \Big( \mathrm{E}[
    \bm{X}_{i,\mathcal{M}} \bm{X}_{i,\mathcal{M}}^T] \Big),~C \cdot \kappa_{\mathrm{max}} \Big( \mathrm{E}[
    \bm{X}_{i,\mathcal{M}} \bm{X}_{i,\mathcal{M}}^T] \Big)\bigg] \\
    &\in [1/C^2, m C^3].
  \end{split}
  \]
  Here we use the fact that the largest eigenvalue of a symmetric matrix is
  upper-bounded by the largest row sum of the matrix. Using the matrix
  Chernoff bound \citep{tropp2012user}, the eigenvalues of
  $((1/n)\tilde{\bm{X}}^T_{\cdot {\mathcal{M}}}\tilde{\bm{X}}_{\cdot
    {\mathcal{M}}})$ are bounded by $1/(2C^2)$ and $2 m C^3$ with
  probability going to $1$.

  The second result follows from
  \[
  \tilde{\bm{\eta}}^T \tilde{\bm{\eta}} = \mathbf{e}_j^T \big(\tilde{\bm{X}}^T_{\cdot
    \hat{\mathcal{M}}}\tilde{\bm{X}}_{\cdot
    \hat{\mathcal{M}}}\big)^{-1} \mathbf{e}_j = \frac{1}{n} \Big[ \big(\frac{1}{n}\tilde{\bm{X}}^T_{\cdot
    \hat{\mathcal{M}}}\tilde{\bm{X}}_{\cdot
    \hat{\mathcal{M}}}\big)^{-1} \Big]_{jj}.
  \]
  The diagonal entries of $((1/n)\tilde{\bm{X}}^T_{\cdot
    \hat{\mathcal{M}}}\tilde{\bm{X}}_{\cdot
    \hat{\mathcal{M}}})^{-1}$ are bounded by its smallest and largest
  eigenvalues, i.e.\ the reciprocal of the largest and smallest
  eigenvalue of $((1/n)\tilde{\bm{X}}^T_{\cdot
    \hat{\mathcal{M}}}\tilde{\bm{X}}_{\cdot \hat{\mathcal{M}}})$.
\end{proof}

Now we turn to the proof of \Cref{lem:beta-tilde-star-random}. By \Cref{assump:model-size},
  $\|{\bm{\beta}}^{*}_{\hat{\mathcal{M}}}\|_{\infty} \le
  \max_{|\mathcal{M}| \le m} \|{\bm{\beta}}^{*}_{{\mathcal{M}}}\|_{\infty}$, $\|\tilde{\bm{\beta}}_{\hat{\mathcal{M}}} -
  {\bm{\beta}}^{*}_{\hat{\mathcal{M}}}\|_{\infty} \le
  \max_{|\mathcal{M}|\le m}\|\tilde{\bm{\beta}}_{{\mathcal{M}}} -
  {\bm{\beta}}^{*}_{{\mathcal{M}}}\|_{\infty}$ with probability
  tending to $1$. By definition,
  \[
  \bm{\beta}^{*}_{{\mathcal{M}}} = \Big[\frac{1}{n} \sum_{i=1}^n
  (T_i - {\mu}_{ti})^2 \bm{X}_{i,\mathcal{M}}\bm{X}_{i,\mathcal{M}}^T\Big]^{-1} \Big[\frac{1}{n} \sum_{i=1}^n
  (T_i - {\mu}_{ti})^2 \Delta(\bm{X}_i)\bm{X}_{i,\mathcal{M}}\Big].
  \]

  By the boundedness of $\mathrm{Var}(T_i|\bm{X}_i)$,
  $\Delta(\bm{X}_i)$ and the uniform boundedness of $\bm{X}_i$,
  \[
  \Big\|\frac{1}{n} \sum_{i=1}^n
  (T_i - {\mu}_{ti})^2
  \Delta(\bm{X}_i)\bm{X}_{i,\mathcal{M}}\Big\|_{\infty} \le C^2 \cdot \frac{1}{n} \sum_{i=1}^n
  (T_i - {\mu}_{ti})^2 = O_p(1).
  \]
  Therefore
  \[
  \begin{split}
    \|\bm{\beta}^{*}_{{\mathcal{M}}}\|_{\infty} &\le \Big\|\Big[\frac{1}{n} \sum_{i=1}^n
    (T_i - {\mu}_{ti})^2
    \bm{X}_{i,\mathcal{M}}\bm{X}_{i,\mathcal{M}}^T\Big]^{-1}\Big\|_1 \cdot \Big\|\frac{1}{n} \sum_{i=1}^n
    (T_i - {\mu}_{ti})^2
    \Delta(\bm{X}_i)\bm{X}_{i,\mathcal{M}}\Big\|_{\infty} \\
    &\le \sqrt{m} \Big\|\Big[\frac{1}{n} \sum_{i=1}^n
    (T_i - {\mu}_{ti})^2
    \bm{X}_{i,\mathcal{M}}\bm{X}_{i,\mathcal{M}}^T\Big]^{-1}\Big\|_2 \cdot C^2 \cdot \frac{1}{n} \sum_{i=1}^n
    (T_i - {\mu}_{ti})^2 \\
    &\le \sqrt{m} C^4 \frac{1}{n} \sum_{i=1}^n (T_i - {\mu}_{ti})^2 = O_p(1).
  \end{split}
  \]
  The last inequality uses \Cref{lem:eta}. Notice that the upper bound
  above holds uniformly for all $|\mathcal{M}| \le m$. Thus
  $\|{\bm{\beta}}^{*}_{\hat{\mathcal{M}}}\|_{\infty} = O_p(1)$.

  For  $\|\tilde{\bm{\beta}}_{{\mathcal{M}}} -
  {\bm{\beta}}^{*}_{{\mathcal{M}}}\|_{\infty}$, by \eqref{eq:psi-beta}
  and the boundedness assumptions (including the boundedness of
  $\bm{\beta}^{*}$), it is easy to show that $\max_{|\mathcal{M}|\le m}\sqrt{n}
  \|\bm{\psi}(\bm{\beta}^{*}_{\mathcal{M}},\hat{\mu}_t)\|_{\infty} =
  o_p(1)$. Therefore by the same argument in the proof of
  \Cref{thm:hat-star}, $\max_{|\mathcal{M}|\le m}\|\tilde{\bm{\beta}}_{{\mathcal{M}}} -
  {\bm{\beta}}^{*}_{{\mathcal{M}}}\|_{\infty} = o_p(n^{-1/2})$.

\subsection{Proof of \Cref{thm:asymptotic-validity}}
\label{sec:proof-crefthm:-valid}

We prove this Theorem through a series of Lemmas.

\begin{lemma}
  \label{lem:replace-lemma-7}
  Under the assumptions of \Cref{thm:beta-tilde-M} and \Cref{assump:support-X},
  \[ \max_{|\mathcal{M}| \le m}
  \|(\tilde{\mathbf{X}}_{\cdot\mathcal{M}}^T
  \tilde{\mathbf{X}}_{\cdot\mathcal{M}})^{-1}
  \tilde{\mathbf{X}}_{\cdot\mathcal{M}}^T (\hat{\bm{\mu}}_y -
  \bm{\mu}_y)\|_{\infty} = o_p(n^{-1/2}).
  \]
\end{lemma}
\begin{proof}
  The proof is similar to the one of \Cref{thm:beta-tilde-M}. For any
  $|\mathcal{M}| \le m$,
  \[
  \begin{split}
    &\Big\|(\tilde{\mathbf{X}}_{\cdot\mathcal{M}}^T
    \tilde{\mathbf{X}}_{\cdot\mathcal{M}})^{-1}
    \tilde{\mathbf{X}}_{\cdot\mathcal{M}}^T (\hat{\bm{\mu}}_y -
    \bm{\mu}_y) \Big\|_{\infty} \\
    =&
    \Big\|\Big[\frac{1}{n}
    \sum_{i=1}^n (T_i - \hat{\mu}_{ti})^2
    \bm{X}_{i,\mathcal{M}}\bm{X}_{i,\mathcal{M}}^T\Big]^{-1} \Big[ \frac{1}{n} \sum_{i=1}^n (T_i -
    \hat{\mu}_{ti}) \mathbf{X}_{i,\mathcal{M}} (\hat{\mu}_{yi} -
    \mu_{yi}) \Big] \Big\|_{\infty} \\
    \le&
    \Big\|\Big[\frac{1}{n}
    \sum_{i=1}^n (T_i - \hat{\mu}_{ti})^2
    \bm{X}_{i,\mathcal{M}}\bm{X}_{i,\mathcal{M}}^T\Big]^{-1}\Big\|_1
    \cdot \Big\|\Big[ \frac{1}{n} \sum_{i=1}^n [(T_i - \mu_{ti}) + (\mu_{ti} - \hat{\mu}_{ti})] \mathbf{X}_{i,\mathcal{M}} (\hat{\mu}_{yi} -
    \mu_{yi}) \Big]\Big\|_{\infty} \\
    \le& \sqrt{m} C^3 \Big| \frac{1}{n} \sum_{i=1}^n [(T_i - \mu_{ti}) + (\mu_{ti} - \hat{\mu}_{ti})] (\hat{\mu}_{yi} - \mu_{yi}) \Big| \\
    =& o_p(n^{-1/2}).
  \end{split}
  \]
  The last inequality uses the rate assumptions in
  \Cref{assump:accuracy-outcome}. As an example, consider
  $(1/n)\sum_{i=1}^n (T_i - \mu_{ti}) (\hat{\mu}_{yi} -
  \mu_{yi})$. Because cross-fitting is used to obtain
  $\hat{\mu}_{yi}$, the two terms in the summand are independent. Thus $\mathbb{E}[(T_i - \mu_{ti}) (\hat{\mu}_{yi} -
  \mu_{yi})] = 0$. Using the consistency of $\hat{\mu}_y$ and the
  Chevyshev inequality, it is straightforward to show $(1/\sqrt{n})\sum_{i=1}^n (T_i - \mu_{ti}) (\hat{\mu}_{yi} -
  \mu_{yi}) = o_p(1)$.
Finally, notice that the bound above is
  universal for all $|\mathcal{M}| \le m$.
\end{proof}

\begin{lemma} \label{lem:kkt}
  Under \Cref{assump:kkt}, $\mathbf{b}_1(\hat{\mathcal{M}},\hat{\mathbf{s}})
  - \mathbf{A}_1(\hat{\mathcal{M}},\hat{\mathbf{s}}) \cdot (\mathbf{Y} -
\hat{\mathbf{\mu}}_y) \ge \bm{1}/(C\sqrt{n})$.
\end{lemma}
\begin{proof}
  By the definition of $\mathbf{A}_1$ and $\mathbf{b}_1$,
\[
\begin{split}
&\mathbf{b}_1(\hat{\mathcal{M}},\hat{\mathbf{s}})
  - \mathbf{A}_1(\hat{\mathcal{M}},\hat{\mathbf{s}}) \cdot (\mathbf{Y} -
\hat{\mathbf{\mu}}_y) \\
=& - \lambda \mathrm{diag}(\mathbf{s})
    (\mathbf{X}_{\cdot,\hat{\mathcal{M}}}^T \mathbf{X}_{\cdot,\hat{\mathcal{M}}})^{-1}
    \mathbf{s} + \mathrm{diag}(\mathbf{s})
    \mathbf{X}_{\cdot,\hat{\mathcal{M}}}^{\dagger} (\mathbf{Y} -
\hat{\mathbf{\mu}}_y)
\end{split}
\]
  The lasso solution $\hat{\bm{\beta}}_{\{1,\dotsc,p\}}(\lambda)$ satisfies
  the Karush-Kuhn-Tucker condition which says that
\[
\mathbf{X}_{\cdot,\hat{\mathcal{M}}}^T \Big[\mathbf{X}_{\cdot,\hat{\mathcal{M}}}
\big(\hat{\bm{\beta}}_{\{1,\dotsc,p\}}(\lambda)\big)_{\hat{\mathcal{M}}} - (\mathbf{Y} -
\hat{\bm{\mu}}_y)\Big] + \lambda \hat{\mathbf{s}} = \mathbf{0}.
\]
Therefore by \Cref{assump:kkt}
\[
\mathbf{b}_1({\hat{\mathcal{M}}},\hat{\mathbf{s}})
  - \mathbf{A}_1(\hat{\mathcal{M}},\hat{\mathbf{s}}) \cdot (\mathbf{Y} -
\hat{\mathbf{\mu}}_y) =
\big|\big(\hat{\bm{\beta}}_{\{1,\dotsc,p\}}(\lambda)\big)_{\hat{\mathcal{M}}}\big|
\ge \bm{1}/(C\sqrt{n}).
\]
\end{proof}

\begin{lemma} \label{lem:denominator}
  Under the assumptions in \Cref{thm:asymptotic-validity}, we have
  \begin{equation} \label{eq:denominator}
  \Phi\left(\frac{U(\mathbf{Y} - \hat{\bm{\mu}}_y;\hat{\mathcal{M}},\hat{\mathbf{s}})-\big(\bm{\beta}^{*}_{\hat{\mathcal{M}}}\big)_j}{\sigma\|\tilde{\bm{\eta}}_{\hat{\mathcal{M}}}\|}\right)
  -
  \Phi\left(\frac{L(\mathbf{Y} - \hat{\bm{\mu}}_y;\hat{\mathcal{M}},\hat{\mathbf{s}})-\big({\bm{\beta}}^{*}_{\hat{\mathcal{M}}}\big)_j}{\sigma\|\tilde{\bm{\eta}}_{\hat{\mathcal{M}}}\|}\right)
  = \Theta_p(1).
  \end{equation}
\end{lemma}
\begin{proof}
By the definition of $U$ and \Cref{lem:eta,lem:replace-lemma-7,lem:beta-tilde-star-random},
\[
\begin{split}
&U(\mathbf{Y} -
\hat{\bm{\mu}}_y;\hat{\mathcal{M}},\hat{\mathbf{s}})-\big({\bm{\beta}}^{*}_{\hat{\mathcal{M}}}\big)_j \\
=& \tilde{\bm{\eta}}_{\hat{\mathcal{M}}}^T (\mathbf{Y} - \hat{\bm{\mu}}_y) -
\big({\bm{\beta}}^{*}_{\hat{\mathcal{M}}}\big)_j + \min_{k:(\mathbf{A}\tilde{\bm{\eta}}_{\hat{\mathcal{M}}})_k > 0}
    \frac{b_k - (\mathbf{A}(\mathbf{Y} - \hat{\bm{\mu}}_y))_k
    }{(\mathbf{A}\tilde{\bm{\eta}}_{\hat{\mathcal{M}}})_k} \\
=& \tilde{\bm{\eta}}_{\hat{\mathcal{M}}}^T \bm{\epsilon} +
\tilde{\bm{\eta}}_{\hat{\mathcal{M}}}^T (\bm{\mu}_y -
\hat{\bm{\mu}}_y) +
\big[\big(\tilde{\bm{\beta}}_{\hat{\mathcal{M}}}\big)_j - \big({\bm{\beta}}^{*}_{\hat{\mathcal{M}}}\big)_j\big] + \min_{k:(\mathbf{A}\tilde{\bm{\eta}}_{\hat{\mathcal{M}}})_k > 0}
    \frac{b_k - (\mathbf{A}(\mathbf{Y} - \hat{\bm{\mu}}_y))_k
    }{(\mathbf{A}\tilde{\bm{\eta}}_{\hat{\mathcal{M}}})_k} \\
\ge& \tilde{\bm{\eta}}_{\hat{\mathcal{M}}}^T \bm{\epsilon} + o_p(1/\sqrt{n}).
\end{split}
\]
The last inequality uses  the KKT conditions \eqref{eq:partition}. Therefore
\[
\frac{U(\mathbf{Y} -
\hat{\bm{\mu}}_y;\hat{\mathcal{M}},\hat{\mathbf{s}})-\big(\tilde{\bm{\beta}}_{\hat{\mathcal{M}}}\big)_j}{\sigma\|\tilde{\bm{\eta}}_{\hat{\mathcal{M}}}\|}
\ge
\Big(\frac{\tilde{\bm{\eta}}_{\hat{\mathcal{M}}}}{\|\tilde{\bm{\eta}}_{\hat{\mathcal{M}}}\|}\Big)^T
\Big(\frac{\bm{\epsilon}}{\sigma}\Big) + o_p(1).
\]
Notice that the first term on the right hand side follows
the standard normal distrbution. Similarly,
\[
\frac{L(\mathbf{Y} -
\hat{\bm{\mu}}_y;\hat{\mathcal{M}},\hat{\mathbf{s}})-\big(\tilde{\bm{\beta}}_{\hat{\mathcal{M}}}\big)_j}{\sigma\|\tilde{\bm{\eta}}_{\hat{\mathcal{M}}}\|}
\le
\Big(\frac{\tilde{\bm{\eta}}_{\hat{\mathcal{M}}}}{\|\tilde{\bm{\eta}}_{\hat{\mathcal{M}}}\|}\Big)^T
\Big(\frac{\bm{\epsilon}}{\sigma}\Big) + o_p(1),
\]
This means the two terms in $\Phi$ in \eqref{eq:denominator} are
not too extreme (the $U$ term cannot go to $-\infty$ and the $L$ term
cannot go to $\infty$). Furthermore, in \Cref{assump:smooth-pivot} it is assumed
that the difference of these two terms is bounded
below. \Cref{eq:denominator} immediate follows from the fact that  the normal
CDF function $\Phi$ has bounded derivative and is lower bounded from $0$ in any
finite interval.
\end{proof}

\begin{lemma} \label{lem:U-difference}
  Under the assumptions in \Cref{thm:asymptotic-validity}, we have
  \[
  \Phi\left(\frac{U(\mathbf{Y} - \bm{\hat{\mu}}_y;\hat{\mathcal{M}},\hat{\mathbf{s}})-\big(\tilde{\bm{\beta}}_{\hat{\mathcal{M}}}\big)_j}{\sigma\|\tilde{\bm{\eta}}_{\hat{\mathcal{M}}}\|}\right)
  -
  \Phi\left(\frac{U(\mathbf{Y} - \bm{\mu}_y;\hat{\mathcal{M}},\hat{\mathbf{s}})-\big(\bm{\beta}^{*}_{\hat{\mathcal{M}}}\big)_j}{\sigma\|\tilde{\bm{\eta}}_{\hat{\mathcal{M}}}\|}\right)
  = o_p(1),
  \]
  and the same conclusion holds for the lower truncation
  threshold $L$.
\end{lemma}
\begin{proof}

  First, we prove an elementary inequality. Suppose $\{b_k\}$ and
  $\{\tilde{b}_k\}$ are two finite sequences of numbers, $b_k \ge 0$ and
  $|\tilde{b}_k - b_k| \le b_k$. Then
  \begin{equation} \label{eq:elementary}
  |\min_k b_k - \min_k \tilde{b}_k| \le (\min_k b_k) \cdot \max_k
  |(\tilde{b}_k/b_k) - 1|.
  \end{equation}
  To prove this, notice that
  \[
  \tilde{b}_k = b_k + b_k ((\tilde{b}_k/b_k) - 1) \ge b_k - b_k \max_k
  |(\tilde{b}_k/b_k) - 1| \ge (\min_k b_k) (1 - \max_k
  |(\tilde{b}_k/b_k) - 1|).
  \]
  Therefore $\min_k \tilde{b}_k - \min_k b_k \ge - (\min_k b_k) \cdot \max_k
  |(\tilde{b}_k/b_k) - 1|$. Conversely, $\min_k \tilde{b}_k - \min_k b_k
  \le \tilde{b}_{k^{*}} - b_{k^{*}} = b_{k^{*}}
  (\tilde{b}_{k^{*}}/b_{k^{*}} - 1) \le b_{k^{*}}
  \max_k|\tilde{b}_{k}/b_{k} - 1|$ where $k^{*} = \argmin_k b_k$.

  Next, we bound the difference between $U(\mathbf{Y} -
  \hat{\bm{\mu}}_y)$ and $U(\mathbf{Y} - \bm{\mu}_y)$. For notational simplicity,
  we suppress the parameters of the selected model
  $(\hat{\mathcal{M}},\hat{\mathbf{s}})$ in $U$ and $\bm{\eta}$.
  \[
  \begin{split}
    &\bigg|\frac{U(\mathbf{Y} - \bm{{\mu}}_y) - U(\mathbf{Y} - \bm{\hat{\mu}}_y)}{\sigma\|\tilde{\bm{\eta}}\|} \bigg| \\
    =&
    \frac{1}{\sigma\|\tilde{\bm{\eta}}\|}
    \bigg|\min_{k:(\mathbf{A}\tilde{\bm{\eta}})_k > 0}
    \frac{b_k - (\mathbf{A}(\mathbf{Y} - \bm{\mu}_y))_k
    }{(\mathbf{A}\tilde{\bm{\eta}})_k} - \min_{k:(\mathbf{A}\tilde{\bm{\eta}})_k > 0}
    \frac{b_k - (\mathbf{A}(\mathbf{Y} - \hat{\bm{\mu}}_y))_k }{(\mathbf{A}\tilde{\bm{\eta}})_k} + \tilde{\bm{\eta}}^T (\hat{\bm{\mu}}_y - {\bm{\mu}}_y) \bigg| \\
    \le&
    \frac{1}{\sigma\|\tilde{\bm{\eta}}\|}
    \bigg|\min_{k:(\mathbf{A}\tilde{\bm{\eta}})_k > 0}
    \frac{b_k - (\mathbf{A}(\mathbf{Y} - \bm{\mu}_y))_k
    }{(\mathbf{A}\tilde{\bm{\eta}})_k} - \min_{k:(\mathbf{A}\tilde{\bm{\eta}})_k > 0}
    \frac{b_k - (\mathbf{A}(\mathbf{Y} - \hat{\bm{\mu}}_y))_k }{(\mathbf{A}\tilde{\bm{\eta}})_k} \bigg| + o_p(1) \\
    \le&
    \frac{1}{\sigma\|\tilde{\bm{\eta}}\|}
    \bigg|\min_{k:(\mathbf{A}\tilde{\bm{\eta}})_k > 0}
    \frac{b_k - (\mathbf{A}(\mathbf{Y} - \hat{\bm{\mu}}_y))_k
    }{(\mathbf{A}\tilde{\bm{\eta}})_k}\bigg| \cdot \max_{k:(\mathbf{A}\tilde{\bm{\eta}})_k > 0} \bigg|\frac{b_k -
      (\mathbf{A}(\mathbf{Y} - {\bm{\mu}}_y))_k}{b_k -
      (\mathbf{A}(\mathbf{Y} - \hat{\bm{\mu}}_y))_k} - 1 \bigg| +
    o_p(1) \\
    =&
    \bigg|\frac{U(\mathbf{Y} - \hat{\bm{{\mu}}}_y)-\tilde{\bm{\eta}}^T(\mathbf{Y} - \hat{\bm{\mu}}_y)}{\sigma\|\tilde{\bm{\eta}}\|}\bigg| \cdot \max_{k:(\mathbf{A}\tilde{\bm{\eta}})_k > 0} \bigg|\frac{b_k -
      (\mathbf{A}(\mathbf{Y} - {\bm{\mu}}_y))_k}{b_k -
      (\mathbf{A}(\mathbf{Y} - \hat{\bm{\mu}}_y))_k} - 1 \bigg| +
    o_p(1) \\
    =&
    \bigg|\frac{U(\mathbf{Y} - \hat{\bm{{\mu}}}_y)-\big(\tilde{\bm{\beta}}_{\hat{\mathcal{M}}}\big)_j
      + \tilde{\bm{\eta}}^T \bm{\epsilon}}{\sigma\|\tilde{\bm{\eta}}\|}\bigg| \cdot \max_{k:(\mathbf{A}\tilde{\bm{\eta}})_k > 0} \bigg|\frac{b_k -
      (\mathbf{A}(\mathbf{Y} - {\bm{\mu}}_y))_k}{b_k -
      (\mathbf{A}(\mathbf{Y} - \hat{\bm{\mu}}_y))_k} - 1 \bigg| +
    o_p(1) \\
    =&
    \bigg|\frac{U(\mathbf{Y} - \hat{\bm{{\mu}}}_y)-\big(\tilde{\bm{\beta}}_{\hat{\mathcal{M}}}\big)_j
    }{\sigma\|\tilde{\bm{\eta}}\|} + O_p(1)\bigg| \cdot \max_{k:(\mathbf{A}\tilde{\bm{\eta}})_k > 0} \bigg|\frac{b_k -
      (\mathbf{A}(\mathbf{Y} - {\bm{\mu}}_y))_k}{b_k -
      (\mathbf{A}(\mathbf{Y} - \hat{\bm{\mu}}_y))_k} - 1 \bigg| + o_p(1) \\
  \end{split}
  \]
  The first inequality uses \Cref{lem:replace-lemma-7} and the second
  inequality uses \eqref{eq:elementary}.

  Using \Cref{lem:kkt,lem:replace-lemma-7}, it is easy to show that
  \[
  \max_{k:(\mathbf{A}\tilde{\bm{\eta}})_k > 0} \bigg|\frac{b_k -
    (\mathbf{A}(\mathbf{Y} - {\bm{\mu}}_y))_k}{b_k -
    (\mathbf{A}(\mathbf{Y} - \hat{\bm{\mu}}_y))_k} - 1 \bigg| = o_p(1).
  \]

  Therefore, using \Cref{lem:beta-tilde-star-random},
  \begin{equation} \label{eq:U-difference}
  \bigg|\frac{U(\mathbf{Y} -
    \hat{\bm{\mu}}_y)-\big(\tilde{\bm{\beta}}_{\hat{\mathcal{M}}}\big)_j}{\sigma\|\tilde{\bm{\eta}}\|}
  -  \frac{U(\mathbf{Y}-\bm{\mu}_y)-\big(\bm{\beta}^{*}_{\hat{\mathcal{M}}}\big)_j}{\sigma\|\tilde{\bm{\eta}}\|}
  \bigg| \le \bigg|\frac{U(\mathbf{Y} - \hat{\bm{{\mu}}}_y)-\big(\tilde{\bm{\beta}}_{\hat{\mathcal{M}}}\big)_j
  }{\sigma\|\tilde{\bm{\eta}}\|} \bigg| \cdot o_p(1) + o_p(1).
  \end{equation}

  Finally, we prove a probability lemma. Let $\{A_n\}$, $\{B_n\}$,
  $\{C_n\}$, $\{D_n\}$ be sequences of random variables such that
  $|A_n - B_n| \le |A_n| C_n + D_n$, $C_n \overset{p}{\to} 0$, $D_n
  \overset{p}{\to} 0$. Then $|\Phi(A_n) - \Phi(B_n)| \overset{p}{\to}
  0$. We prove this result for deterministic sequences (in probability
  convergence is changed to determinstic limit). We only need to
  prove the result for two infinite subsequences of $\{A_n\}$, $\{A_n:A_n\le
  1\}$ and $\{A_n:A_n>1\}$ (if any subsequence is finite then we can
  ignore it). Within the first subsequence, we have $|A_n - B_n| \to 0$ and hence
  $\Phi(A_n) - \Phi(B_n) \to 0$. Within the second subsequence, for
  large enough $n$ we   have $|A_n - B_n| \le A_n/2$, so $|\Phi(A_n) -
  \Phi(B_n)| \le \max(\phi(A_n),\phi(B_n)) |A_n - B_n| \le
  \phi(A_n/2) (|A_n| C_n + D_n) \to 0$, where we
  have used the fact that $\phi(ca) a$ is a bounded function of $a \in
  [1,\infty]$ for any constant $c > 0$.

  Using \eqref{eq:U-difference} and the result above, we have
  \[
  \begin{split}
    &\bigg|\Phi\left(\frac{U(\mathbf{Y} - \hat{\bm{\mu}}_y)-\big(\tilde{\bm{\beta}}_{\hat{\mathcal{M}}}\big)_j}{\sigma\|\tilde{\bm{\eta}}\|}\right)
    -
    \Phi\left(\frac{U(\mathbf{Y}-\bm{\mu}_y)-\big(\bm{\beta}^{*}_{\hat{\mathcal{M}}}\big)_j}{\sigma\|\tilde{\bm{\eta}}\|}\right)
    \bigg| \overset{p}{\to} 0 \\
  \end{split}
  \]
  as desired.
\end{proof}

Finally we turn to the proof of \Cref{thm:asymptotic-validity}. By
\Cref{lem:lasso-selective-inference}, we have, conditioning on the
event $\big\{\hat{\mathcal{M}}_{\lambda}(\mathbf{Y} - \bm{\mu}_y) =
\mathcal{M}, \hat{\mathbf{s}}_{\lambda}(\mathbf{Y}-\bm{\mu}_y) =
\mathbf{s}\big\}$,
\begin{equation} \label{eq:lasso-pivot-new}
  F\Big(\big(\hat{\bm{\beta}}_{{\mathcal{M}}}(\mathbf{Y} - \bm{\mu}_y)\big)_j;
  \big(\tilde{\bm{\beta}}_{{\mathcal{M}}}\big)_j, \sigma^2 \tilde{\bm{\eta}}_{\mathcal{M}}^T
  \tilde{\bm{\eta}}_{\mathcal{M}},L\big(\mathbf{Y} - \bm{\mu}_y;{\mathcal{M}},{\mathbf{s}}\big),U\big(\mathbf{Y} - \bm{\mu}_y;{\mathcal{M}},{\mathbf{s}}\big)
  \Big) \sim \mathrm{Unif}(0,1),
\end{equation}
To prove \Cref{thm:asymptotic-validity}, we just need to replace
$\bm{\mu}_y$ by $\hat{\bm{\mu}}_y$ and $\tilde{\bm{\beta}}$ by
$\bm{\beta}^{*}$ in the above equation and prove convergence in
distribution. We begin by showing $\mathbf{P}\big(\hat{\mathcal{M}}_{\lambda}(\mathbf{Y} -
\bm{\mu}_y) = \hat{\mathcal{M}}_{\lambda}(\mathbf{Y} -
\hat{\bm{\mu}}_y), \hat{\mathbf{s}}_{\lambda}(\mathbf{Y} -
\bm{\mu}_y) = \hat{\mathbf{s}}_{\lambda}(\mathbf{Y} -
\hat{\bm{\mu}}_y)\big) \to 1$. Recall that the event
$\{\hat{\mathcal{M}}_{\lambda}(\mathbf{Y}) = \mathcal{M},
\hat{\mathbf{s}}_{\lambda}(\mathbf{Y}) = \mathbf{s}\}$ is
characterized by the KKT conditions $\big\{\mathbf{A}_1(\mathcal{M},\mathbf{s}) \bm{Y} \le
\mathbf{b}_1(\mathcal{M},\mathbf{s})\big\}$. By \Cref{lem:kkt}, with
probability tending to $1$ these inequalities are satisfied with a
margin at least $1/(C\sqrt{n})$ for $\mathbf{Y} - \hat{\mu}_y$. By
\Cref{lem:replace-lemma-7}, $\|\mathbf{A}_1(\bm{\mu}_y -
\hat{\bm{\mu}}_y)\|_{\infty} = o_p(n^{-1/2})$ and hence the KKT
conditions are also satisfied for $\mathbf{Y} - \bm{\mu}_y$
with probability tending to $1$. Next, let's write down the pivotal statistic in
\eqref{eq:asymptotic-validity}
\[
\begin{split}
&F\Big(\big(\hat{\bm{\beta}}_{{\hat{\mathcal{M}}}}\big)_j;
  \big({\bm{\beta}}^{*}_{{\hat{\mathcal{M}}}}\big)_j, \sigma^2 \tilde{\bm{\eta}}_{\hat{\mathcal{M}}}^T
  \tilde{\bm{\eta}}_{\hat{\mathcal{M}}},L\big(\mathbf{Y} - \hat{\bm{\mu}}_y;{\hat{\mathcal{M}}},{\mathbf{s}}\big),U\big(\mathbf{Y} - \hat{\bm{\mu}}_y;{\hat{\mathcal{M}}},{\mathbf{s}}\big)
  \Big) \\
&= \frac{\Phi\left(\frac{\big(\hat{\bm{\beta}}_{\hat{\mathcal{M}}}\big)_j-\big({\bm{\beta}}^{*}_{\hat{\mathcal{M}}}\big)_j}{\sigma\|\tilde{\bm{\eta}}\|}\right)
  - \Phi\left(\frac{L(\mathbf{Y} - \hat{\bm{{\mu}}}_y;\hat{\mathcal{M}},\hat{\mathbf{s}})-\big({\bm{\beta}}^{*}_{\hat{\mathcal{M}}}\big)_j}{\sigma\|\tilde{\bm{\eta}}\|}\right)}{  \Phi\left(\frac{U(\mathbf{Y} - \hat{\bm{{\mu}}}_y;\hat{\mathcal{M}},\hat{\mathbf{s}})-\big({\bm{\beta}}^{*}_{\hat{\mathcal{M}}}\big)_j}{\sigma\|\tilde{\bm{\eta}}\|}\right)
  -
  \Phi\left(\frac{L(\mathbf{Y} - \hat{\bm{{\mu}}}_y;\hat{\mathcal{M}},\hat{\mathbf{s}})-\big({\bm{\beta}}^{*}_{\hat{\mathcal{M}}}\big)_j}{\sigma\|\tilde{\bm{\eta}}\|}\right)
}.
\end{split}
\]
By \Cref{lem:denominator}, the denominator of the right hand side is
$\Theta_p(1)$. Therefore using
\Cref{lem:beta-tilde-star-random,lem:replace-lemma-7,lem:U-difference},
we can replace
$\bm{\mu}_y$ by $\hat{\bm{\mu}}_y$ and $\tilde{\bm{\beta}}$ by
$\bm{\beta}^{*}$ in the numerator of the right hand side and show the
difference is $o_p(1)$. Now using \Cref{lem:U-difference}, we can replace
$\bm{\mu}_y$ by $\hat{\bm{\mu}}_y$ and $\tilde{\bm{\beta}}$ by
$\bm{\beta}^{*}$ in the numerator and show the difference again is
$o_p(1)$. Therefore we have proved that
\[
\begin{split}
  &F\Big(\big(\hat{\bm{\beta}}_{{\hat{\mathcal{M}}}}\big)_j;
  \big(\tilde{\bm{\beta}}_{{\hat{\mathcal{M}}}}\big)_j, \sigma^2 \tilde{\bm{\eta}}_{\hat{\mathcal{M}}}^T
  \tilde{\bm{\eta}}_{\hat{\mathcal{M}}},L\big(\mathbf{Y} - \bm{\mu}_y;{\hat{\mathcal{M}}},{\hat{\mathbf{s}}}\big),U\big(\mathbf{Y} - \bm{\mu}_y;{\hat{\mathcal{M}}},{\hat{\mathbf{s}}}\big)
  \Big) \\
  &-   F\Big(\big(\hat{\bm{\beta}}_{\hat{\mathcal{M}}}\big)_j;
  \big(\bm{\beta}^{*}_{\hat{\mathcal{M}}}\big)_j, \sigma^2 \tilde{\bm{\eta}}_{\hat{\mathcal{M}}}^T
  \tilde{\bm{\eta}}_{\hat{\mathcal{M}}},L\big(\mathbf{Y} - \hat{\bm{\mu}}_y;\hat{\mathcal{M}},\hat{\mathbf{s}}\big),U\big(\mathbf{Y} - \hat{\bm{\mu}}_y;\hat{\mathcal{M}},\hat{\mathbf{s}}\big)
  \Big) = o_p(1).
\end{split}
\]
Combining this with \eqref{eq:lasso-pivot-new}, we have thus proved the main
Theorem.




\end{document}